\documentclass[journal]{IEEEtran}
\newcounter{boxlblcounter}  

\newcommand{\avec}{{\bf{a}}}

\newcommand{\bvec}{{\bf{b}}}

\newcommand{\zerovec}{{\bf{0}}}

\newcommand{\alphavec}{{\bf{\alpha}}}

\newcommand{\Gammamat}{{\bf{\Gamma}}}
\newcommand{\Amat}{{\bf{A}}}

\newcommand{\Cmat}{{\bf{C}}}
\newcommand{\Dmat}{{\bf{D}}}
\newcommand{\Emat}{{\bf{E}}}

\newcommand{\Hmat}{{\bf{H}}}

\newcommand{\Imat}{{\bf{I}}}
\newcommand{\Lmat}{{\bf{L}}}

\newcommand{\Vmat}{{\bf{V}}}

\newcommand{\Ymat}{{\bf{Y}}}

\newcommand{\w}{{\rm{w}}}

\newcommand{\define}{\stackrel{\triangle}{=}}

\def\alphavec{{\mbox{\boldmath $\alpha$}}}

\newcommand{\be}{\begin{equation}}
\newcommand{\ee}{\end{equation}}
\newcommand{\beqna}{\begin{eqnarray}}
\newcommand{\eeqna}{\end{eqnarray}}

\newcommand{\mySet}[1]{\mathcal{#1}}

\newcommand{\Data}{\mySet{D}}
\newcommand{\E}{{\mathbb{E}}}
\newcommand{\Ntraining}{K}
\usepackage{amsfonts,amsthm,amsmath,amssymb}
\usepackage{graphicx,subcaption}
\usepackage[noadjust]{cite}
\usepackage[dvipsnames]{xcolor} 
\DeclareMathAlphabet{\pazocal}{OMS}{zplm}{m}{n}
\DeclareMathOperator*{\argmin}{argmin}
\DeclareMathOperator*{\argmax}{argmax}
\DeclareMathOperator{\EX}{\mathbb{E}}
\usepackage{transparent}
\newtheorem{Claim}{Claim}

\newtheorem{Theorem}{Theorem}

\usepackage{acronym}
\acrodef{map}[MAP]{maximum a-posteriori probability}
\acrodef{snr}[SNR]{signal-to-noise ratio}
\acrodef{mse}[MSE]{mean-squared-error}
\acrodef{mse_abs}[MSE]{mean-squared error}
\acrodef{wlmmse_abs}[WLMMSE]{widely-linear minimum mean-squared-error}
\acrodef{wlmmse}[WLMMSE]{widely-linear minimum mean-squared-error}
\usepackage{latexsym,amsmath,amsthm,amssymb,epsfig,enumerate,color,graphicx,cite}
\usepackage[lined,algonl,boxed]{algorithm2e}
\usepackage{tikz}
\usetikzlibrary{shapes,arrows}
\usepackage{verbatim} 
\usepackage{pgfplots}

\usepackage{soul}

\tikzset{%
	block/.style    = {draw, thick, rectangle, minimum height = 3em,
		minimum width = 3em},
	mult/.style    = {draw, circle, node distance = 2cm}, 
	sum/.style      = {draw, circle, node distance = 2cm}, 
	input/.style    = {coordinate}, 
	output/.style   = {coordinate} 
	label/.style   = {draw=none} 
	
}
\newcommand{\suma}{\Large$+$}

\begin{document}
\title{Widely-Linear MMSE Estimation of Complex-Valued Graph Signals} 
\author{Alon~Amar and Tirza Routtenberg, \IEEEmembership{Senior Member, IEEE }
	\thanks{A. Amar is an adjunct instructor at the Faculty of Electrical \& Computer Engineering, Technion, Haifa, Israel.
	T. Routtenberg is with the School of Electrical and Computer Engineering, Ben-Gurion University of the Negev, Beer Sheva, Israel and with the Electrical and Computer Engineering Department, Princeton University, Princeton, NJ. 
e-mail: aamar@technion.ac.il, tirzar@bgu.ac.il.
This work is partially supported by the  Israeli Ministry of National Infrastructure, Energy, and Water Resources and by the ISRAEL SCIENCE FOUNDATION (grant No. 1148/22).\\
$\copyright$ 2023 IEEE.  Personal use of this material is permitted.  Permission from IEEE must be obtained for all other uses, in any current or future media, including reprinting/republishing this material for advertising or promotional purposes, creating new collective works, for resale or redistribution to servers or lists, or reuse of any copyrighted component of this work in other works.}
}


\maketitle

\def\E{\textrm{E}}
\def\blockdiag{\text{block-diag}}
\def\tr{\text{tr}}
\def\diag{\text{diag}}
\def\x{{\mathbf x}}
\def\L{\mathbf{L}}
\def\w{\mathbf{w}}
\def\W{\mathbf{W}}
\def\X{\mathbf{X}}
\def\Y{\mathbf{Y}}
\def\Z{\mathbf{Z}}
\def\Emat{\mathbf{E}}
\def\What{\widehat{\mathbf{W}}}
\def\A{\mathbf{A}}
\def\Aqbar{\mathbf{\bar{A}}_{q}}
\def\PAqbar{\mathbf{P}_{\mathbf{\bar{A}}_{q}}}
\def\PAqbarOrth{\mathbf{P}_{\mathbf{\bar{A}}_{q}}^{\bot}}
\def\p{\mathbf{p}}
\def\phat{\hat{\mathbf{p}}}
\def\Gammaq{\mathbf{\Gamma}_{q}}
\def\H{\mathbf{H}}
\def\V{\mathbf{V}}
\def\I{\mathbf{I}}
\def\J{\mathbf{J}}
\def\bq{\mathbf{b}_{q}}
\def\aq{\mathbf{a}_{q}}
\def\d{\mathbf{d}}
\def\u{\mathbf{u}}
\def\v{\mathbf{v}}
\def\g{\mathbf{g}}
\def\D{\mathbf{D}}
\def\R{\mathbf{R}}
\def\a{\mathbf{a}}
\def\q{\mathbf{q}}
\def\C{\mathbf{C}}
\def\bqhat{\hat{\mathbf{b}}_{q}}
\def\ahat{\hat{\mathbf{a}}}
\def\aqhat{\hat{\mathbf{a}}_{q}}
\def\alq{\mathbf{a}_{\ell q}}
\def\alqhat{\hat{\mathbf{a}}_{\ell q}}
\def\Rss{\mathbf{R}_{\mathbf{ss}}}
\def\Rrr{\mathbf{R}_{\mathbf{rr}}}
\def\Rsr{\mathbf{R}_{\mathbf{sr}}}
\def\Ahat{\hat{\mathbf{A}}}
\def\n{\mathbf{n}}
\def\s{\mathbf{s}}
\def\J{\mathbf{J}}
\def\T{\mathbf{T}}
\def\b{\mathbf{b}}
\def\e{\mathbf{e}}
\def\r{\mathbf{r}}
\def\Lam{\mathbf{\Lambda}}
\def\Lambdas{\mathbf{\Lambda}_{\tilde{\s}}}
\def\F{{\mathbf{F}}}
\def\Fdot{\dot{\mathbf{F}}}
\def\Fddot{\ddot{\mathbf{F}}}
\def\Gq{\mathbf{G}_{q}}
\def\G{\mathbf{G}}
\def\O{\mathbf{O}}
\def\f{\mathbf{f}}
\def\g{\mathbf{g}}
\def\h{\mathbf{h}}
\def\R{\mathbf{R}}
\def\D{\mathbf{D}}
\def\Q{\mathbf{Q}}
\def\P{\mathbf{P}}
\def\C{\mathbf{C}}
\def\c{\mathbf{c}}
\def\B{\mathbf{B}}
\def\X{\mathbf{X}}
\def\T{\mathbf{T}}
\def\S{\mathbf{S}}
\def\s{\mathbf{s}}
\def\A{\mathbf{A}}
\def\P{\mathbf{P}}
\def\p{\mathbf{p}}
\def\t{\mathbf{t}}
\def\z{\mathbf{z}}
\def\Z{\mathbf{Z}}
\def\e{\mathbf{e}}
\def\m{\mathbf{m}}
\def\N{\mathbf{N}}
\def\U{\mathbf{U}}
\def\X{\mathbf{X}}
\def\Y{\mathbf{Y}}
\def\y{\mathbf{y}}
\def\Ups{\mathbf{\Upsilon}}
\def\ESNR{\text{ESNR}}
\def\SNR{\text{SNR}}
\def\min{\text{min}}
\def\max{\text{max}}
\def\dim{\text{dim}}
\def\Var{\text{Var}}
\def\Null{\text{Null}}
\def\rank{\text{rank}}
\def\RMSE{\text{RMSE}}
\def\MSE{\text{MSE}}
\def\cov{\text{cov}}
\def\max{\text{max}}
\def\argmax{\text{argmax}}
\def\argmin{\text{argmin}}
\def\log{\text{log}}
\def\ln{\text{ln}}
\def\vec{\text{vec}}
\def\BD{\text{BlockDiag}}
\def\rh{\mathbf{\rho}}
\def\thet{\mathbf{\theta}}
\def\Sigm{\mathbf{\Sigma}}
\def\exp{\text{exp}}
\def\det{\text{det}}
\def\K{\text{K}}
\def\Diag{\text{Diag}}
\def\N{\text{N}}
\def\M{\text{M}}
\def\RC{\text{RC}}
\def\FIM{\text{FIM}}
\def\sinc{\text{sinc}}
\def\eqdef{\buildrel\Delta\over=}
\newcommand{\ve}[1]{\boldsymbol{#1}}
\theoremstyle{remark}
\newtheorem*{remark}{Remark\normalfont} 

 \begin{abstract}
 In this paper, we consider the problem of recovering  random graph signals with complex values. For general Bayesian estimation of complex-valued vectors, it is known that the  widely-linear minimum mean-squared-error (WLMMSE) estimator can achieve a lower mean-squared-error (MSE) than that of the linear minimum MSE (LMMSE) estimator for the estimation of improper complex-valued signals. Inspired by the WLMMSE estimator, in this paper  we develop the graph signal processing (GSP)-WLMMSE estimator, which minimizes the MSE among estimators that are represented as a two-channel output of a graph filter, i.e. widely-linear GSP estimators. We discuss the properties of the proposed GSP-WLMMSE estimator. In particular, we show that the MSE of the GSP-WLMMSE estimator is always equal to or lower than the MSE of the GSP-LMMSE estimator. The GSP-WLMMSE estimator is based on {\em{diagonal}} covariance matrices in the graph frequency domain, and thus has reduced complexity compared with the WLMMSE estimator. This property is especially important when using the sample-mean versions of these estimators that are based on a training dataset.
We state conditions under which the low-complexity GSP-WLMMSE estimator coincides with the WLMMSE estimator.
In simulations, we investigate a synthetic linear estimation problem and the nonlinear problem of  state estimation in power systems.  For these problems, it is shown that the GSP-WLMMSE estimator outperforms the GSP-LMMSE estimator and achieves similar performance to that of the WLMMSE estimator. Moreover, 
the sample-mean version of the GSP-WLMMSE estimator outperforms the sample-mean WLMMSE estimator for a limited training dataset and is  more robust to topology changes.
\end{abstract}


\section{Introduction}
Graph signal processing (GSP)  theory extends concepts and techniques from traditional digital signal processing (DSP) to data indexed by generic graphs, including the graph Fourier transform (GFT), graph filter design \cite{8347162,Shuman_Ortega_2013}, and sampling  of graph signals \cite{9244650,6854325}.
GSP theory provides a set of powerful tools for processing data over networks in terms of computational efficiency, signal processing in a decentralized manner, utilization of relevant network information by employing the graph filtering framework, and more \cite{8347162,Shuman_Ortega_2013}. 
In this context, the development of GSP methods for the 
recovery (or estimation) of graph signals is a fundamental
task  with both practical and theoretical importance \cite{routtenberg2020,kroizer2021bayesian}.
However, 
most  works have focused on real-valued graph signals.
Complex-valued graph signals arise in several real-world applications, such as multi-agent systems  \cite{han2015formation},
 wireless communication systems \cite{Tscherkaschin_2016},
 voltage and power phasors in electrical networks  \cite{2021Anna,drayer2018detection,dabush2021state,Morgenstern_2022}, and  probabilistic graphical models with 
complex-valued multivariate Gaussian vectors \cite{Tugnait_2019}.
Despite the widespread use of complex-valued graph signals, 
 their recovery is a crucial problem that has not been well explored.

In    processing general complex-valued random vectors, standard statistical signal processing approaches typically impose a restrictive assumption on these vectors in the form of properness or second-order circularity  \cite{neeser1993proper}.
 However, this assumption does not necessarily hold, and in various applications the  data  is noncircular or improper (i.e. the signals have linear correlations with their own complex conjugates \cite{schreier2010statistical}). This situation arises, for example, through power imbalance
and correlation in data channels \cite{nitta2019hypercomplex,Yili_mandic}. 
Various tools for
the processing of  noncircular complex-valued signals  for different tasks were proposed  in the literature  \cite{eriksson2006complex,li2008complex,Picinbono_Chevalier1995,JAHANCHAHI201433,boloix2017widely}, based on 
 augmented complex statistics \cite{schreier2010statistical,mandic2009complex}.  
 In the context of parameter estimation, it has been shown that
if the observations and/or
the parameters of interest are improper,  widely-linear estimators can outperform conventional  linear estimators.
In particular,
the \ac{wlmmse}  estimator, first proposed in \cite{Picinbono_Chevalier1995}, significantly outperforms  the linear minimum mean-squared-error (LMMSE) estimator for improper random vectors and has many desired properties \cite{schreier2010statistical}.
However, conventional widely-linear estimation does not utilize the graph properties to obtain efficient estimation. Moreover, for the general nonlinear case, the covariance matrices used in the \ac{wlmmse} are difficult to determine. Thus, estimating these matrices from data with high accuracy is needed, which often necessitates a larger sample size than is available \cite{Eldar_merhav}. To address these challenges, obtain efficient estimation, and benefit from the powerful GSP tools, there is a need  for widely-linear estimation methods for improper complex-valued graph signals based on GSP.

Graph filters have been used for many signal processing tasks, such as classification \cite{6778068}, denoising \cite{7032244},  and anomaly detection \cite{drayer2018detection}. Model-based recovery of a graph signal from noisy measurements by graph filters for {\em{linear}} models, usually with {\em{real-valued}} data, was treated in \cite{7117446,7032244,Isufi_Leus2017,7891646}.
The graphical  Wiener filter was developed
in  \cite{7891646} and in \cite{Yagan_Ozen_2020}  for real and for proper complex signals, respectively.
Nonlinear methods, such as the nonlinear graph filters that were considered in \cite{8496842},  require higher-order statistics that are not completely specified in the general case, and have  significantly greater computational complexity.
Graph neural network approaches  \cite{Isufi_Ribeiro,gama2020graphs} require extensive training sets and result in nonlinear estimators with high computational complexity.
Fitting graph-based models to given data was considered in \cite{7763882,Hua_Sayed_2020,confPaper}. Based on the fitted
 graph-based model, one can implement different simple recovery algorithms.
However, model-fitting approaches aim to minimize the modeling error and, in general, have significantly lower performance  than  estimators that minimize the estimation error directly.

The GSP-linear minimum mean-squared-error (GSP-LMMSE) estimator was proposed in our work in \cite{kroizer2021bayesian}  for the recovery of {\em{real-valued}} graph signals. The GSP-LMMSE estimator minimizes the \ac{mse} among the subset of estimators that are represented as an output of a graph filter.
Compared with the LMMSE estimator, the GSP-LMMSE estimator can achieve significant complexity reduction without seriously compromising performance. Moreover, it requires less training data when there is a need to estimate the covariance matrices \cite{kroizer2021bayesian}. 
It also has graph filter implementations that are adaptive to changes in the graph topology,  outperforming the LMMSE estimator under such changes. However, it may be inappropriate for processing {\em{improper}} complex-valued graph signals. 
Our goal in this paper is  to 
efficiently exploit the advantages provided by the GSP framework to
obtain a low-complexity estimator 
that has optimal \ac{mse} performance for 
the recovery of improper  {\em{complex-valued}}  graph signals.
This is illustrated in the diagram of the relationships between the different classes of estimators in  Fig. \ref{diagram}. Thus, we aim to fill the gap and derive GSP-widely-linear estimators that display better MSE performance than the GSP-linear estimators, but with the advantages of GSP (robustness, low computational complexity, efficient distributed implementation, accuracy under mismatches and estimated parameters, etc. as outlined in \cite{kroizer2021bayesian}) that do not exist in the class of conventional linear/widely-linear estimators. 
\begin{figure}
\centerline{\includegraphics[width=0.5\columnwidth]{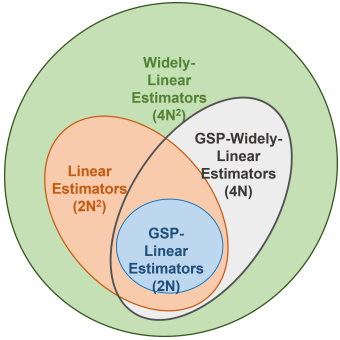}}
 	\caption{The relationships between the classes of estimators: linear, widely-linear, GSP-linear, and GSP-widely-linear estimators. The numbers in the brackets ($4 N^2$, $4N$, $2 N^2$, and $2N$) are the degrees of freedom, i.e. the number of real parameters that we need to set in order to find the estimators for the recovery of  an $N$-dimensional graph signal.
}
 	\label{diagram}
  \vspace{-0.5cm}
\end{figure}

 In this paper, we describe a Bayesian framework for the widely-linear estimation  of complex-valued graph signals by using  GSP filters.  Thus, the proposed framework takes advantage of the power of both  complex-domain nonlinear processing and GSP theory.
 In particular, we develop the GSP-\ac{wlmmse} estimator that achieves the minimum \ac{mse} within the class of widely-linear  estimators that are the output of graph filters.
 We show that using widely-linear graph filtering instead of strictly linear graph filtering  can yield significant \ac{mse} improvements in
the estimation of graph signals with complex data.
We discuss more aspects of the proposed GSP-WLMMSE estimator, including the orthogonality principle, computational complexity, relation with existing estimators, and its behavior for the special cases of proper signals and data, maximal improper observation vectors, and real-valued unknown graph signals. 
 In simulations, it is shown that the GSP-\ac{wlmmse} estimator has a significantly lower \ac{mse} than that of the LMMSE estimator for a synthetic example and for  the problem of state estimation in power systems. 

 \textit{Organization:} The rest of this paper is organized as follows. In Section \ref{background_sec} we introduce the basics of GSP and widely-linear estimation required for this paper.
In Section \ref{new_estimator_sec}, we formulate the estimation problem and develop the proposed GSP-WLMMSE estimator for improper graph signals and discuss its properties. 
Some special cases are presented in Section \ref{special_cases_subsec}.
Simulations are shown in Section \ref{simulation_sec}.
Finally, the paper is concluded in Section \ref{conc}.

 \textit{Notation:} We use boldface lowercase letters to denote vectors and boldface capital letters for matrices. The identity matrix is denoted by $\Imat$
 and  $||\cdot||$ denotes the
Euclidean $l_2$-norm of vectors. The symbols $(\cdot)^*$,$(\cdot)^T$, and $(\cdot)^H$ represent the conjugate, transpose, and conjugate transpose operators, respectively. The symbol $(\cdot)^{-*}$ denotes the  conjugate of the inverse matrix.
 We use $ \tr(\Amat)$ to denote the trace of the  matrix $\Amat$.
 For a vector $\avec$, ${\text{diag}}(\avec)$ is a diagonal matrix whose $i$th diagonal entry is $a_i$; when applied to a matrix, ${\text{diag}}(\Amat)$ is a vector collecting the diagonal elements of $\Amat$. In addition, ${\text{ddiag}}(\Amat)={\text{diag}}({\text{diag}}(\Amat))$.
 The notations $\Re\{\cdot\}$ and $\Im\{\cdot\}$ denote the real and imaginary parts of their argument.
Two zero-mean random vectors $\a,\b\in \mathbb{C}^{N}$ are defined as  orthogonal, $\a \perp \b$, if $\EX[\avec^H \bvec]=0$.
For two zero-mean complex-valued vectors $\avec$ and $\bvec$,
we denote the
cross-covariance and the complementary cross-covariance  matrices  between them  as
\begin{equation}
\label{one}
\Gammamat_{\avec\bvec}=\EX[\avec \bvec^H]~~~{\text{and}}~~~
\C_{\avec\bvec} = \EX[\avec\bvec^T],
\end{equation}
respectively.  The complementary cross-covariance matrix, $\C_{\avec\bvec}$, is 
 also known as the pseudo or conjugate covariance.
\section{Background}
\label{background_sec}
In this section, we briefly present  background on  \ac{wlmmse} estimation (Subsection \ref{WLMMSE_sub}),
the 
GSP notations  (Subsection \ref{GSP_subsec}), and the GSP-LMMSE estimator (Subsection \ref{GSP_LMMSE_subsec}).
\subsection{WLMMSE Estimation}
\label{WLMMSE_sub}
Consider the problem of estimating an $N$-dimensional zero-mean  signal $\x\in \mathbb{C}^{N} $ from an $N$-dimensional zero-mean measurement $\y\in \mathbb{C}^{N} $.
 In order for an estimator to utilize all the available second-order information for proper and improper signals, the estimators should depend on both the signals  and their complex conjugates, i.e. on
the augmented vector $
    \underline{\y} =
    \left[
        \y^T~
        \y^H
    \right]^T$. 
It can be seen that the augmented auto-covariance matrix satisfies
\begin{equation}
\label{augmented_cov}
   {\Gammamat}_{ \underline{\y}  \underline{\y}} = \EX[\underline{\y}\,\underline{\y}^H] 
   = 
    \left[
        \begin{array}{cc}
           \Gammamat_{\y\y}   &  \C_{\y\y}\\
           \C_{\y\y}^*  &  \Gammamat_{\y\y}^*
        \end{array}
    \right]
    =
    {\Gammamat}_{\underline{\y}\underline{\y}}^H,
\end{equation}
where the cross-covariance  and the complementary cross-covariance  matrices  are $\Gammamat_{\y\y}=\EX[\y \y^H]$ and  $ \C_{\y\y}= \EX[\y\y^T]$, as defined in \eqref{one}.
If the complementary covariance matrix vanishes, i.e. $ \C_{\y\y} = \ve{0}$, the vector $\y$ is called a proper vector, while when $ \C_{\y\y}  \neq \ve{0}$ it is called an improper vector.  We assume in this paper that $\Gammamat_{\y\y}$ is a non-singular matrix.

The  LMMSE estimator of the signal $\x$ based on the measurement vector $\y$, for zero-mean vectors $\x$ and $\y$, is \cite{Kayestimation}
\begin{equation}
\label{LMMSE}
   \hat\x^{\text{LMMSE}}\define \Gammamat_{\x\y}\Gammamat_{\y\y}^{-1}\y.
\end{equation}
In addition, 
the \ac{mse}  of the LMMSE estimator from \eqref{LMMSE} is
\begin{equation}
\label{MSE_LMMSE}
    \varepsilon_{L}^2
    \define
    \EX[\|\hat{\x}^{\text{LMMSE}} - \x\|^2]  
    =
    \tr
    (
    \Gammamat_{\x\x} 
    - 
    \Gammamat_{\x\y}\Gammamat_{\y\y}^{-1}\Gammamat_{\x\y}^H
    ).
    \end{equation} 
For improper random variables, the widely-linear estimators can be used to achieve better performance than that of the LMMSE estimator. The widely-linear transformation provides a convenient representation, which enables the incorporation of full second-order statistical information into the estimation scheme and provides significant advantages when the signal is improper \cite{schreier2010statistical,Schreier_Scharf_2003,mandic2009complex,Picinbono_Chevalier1995}.
A widely-linear estimator of the signal $\x$ based on the measurement vector $\y$ has the form
\begin{equation}
\label{general_form_widely_lin}
    \hat{\x} = \H_1 \y + \H_2 \y^*. 
\end{equation}
In order to find the WLMMSE estimator, 
the filters $\H_1 ,\H_2 \in \mathbb{C}^{N \times N}$ are determined by minimizing the \ac{mse} as follows:
\begin{eqnarray}
\label{H1H2minimization}
    \underset{\H_1,\H_2\in \mathbb{C}^{N \times N}}\min\,
    \EX[\|\hat{\x} - \x\|^2] \hspace{3cm}\nonumber\\
    = \underset{\H_1,\H_2\in \mathbb{C}^{N \times N}}\min\,
    \EX[\| \H_1 \y + \H_2 \y^*  - \x\|^2].
\end{eqnarray}
The solution of  \eqref{H1H2minimization} is given by \cite{Picinbono_Chevalier1995}
\begin{eqnarray}
\label{H1_optimal_WLMMSE}
    \hat\H_1 &=& (\Gammamat_{\x\y} - {\C}_{\x\y} (\Gammamat_{\y\y}^{-1})^* ({\C}_{\y\y})^*) \P_{\y\y}^{-1} \\
    \label{H2_optimal_WLMMSE}
     \hat\H_2 &=&
    (\C_{\x\y} - \Gammamat_{\x\y} \Gammamat_{\y\y}^{-1} {\C}_{\y\y}) (\P_{\y\y}^{-1})^*,
\end{eqnarray}
where the Schur complement of the matrix $  {\Gammamat}_{ \underline{\y} \underline{\y}}$, which is the error covariance matrix for linearly estimating $\y$ from $\y^*$, is 
\begin{equation}
\label{Pyy}
    \P_{\y\y} \define \Gammamat_{\y\y} -  \C_{\y\y} (\Gammamat_{\y\y}^{-1})^* ({\C}_{\y\y})^*.
\end{equation}
By substituting \eqref{H1_optimal_WLMMSE} and \eqref{H2_optimal_WLMMSE} in \eqref{general_form_widely_lin}, we obtain that the WLMMSE estimator is
\beqna
\label{WLMSE_est}
    \hat\x^{\text{WLMMSE}} &=& (\Gammamat_{\x\y}-{\C}_{\x\y}(\Gammamat_{\y\y}^{-1})^*\C_{\y\y}^*)\P_{\y\y}^{-1}\y \nonumber\\&&
    + 
    (\C_{\x\y}-\Gammamat_{\x\y}\Gammamat_{\y\y}^{-1}{\C}_{\y\y})(\P_{\y\y}^{-1})^*\y^*. 
\eeqna
By using the fact that 
 $\EX[\y \y^H]= \Gammamat_{\y\y}$ and  $  \EX[\y\y^T]=\C_{\y\y}$, it can be verified that 
the corresponding \ac{mse} of   the WLMMSE estimator is \cite{Picinbono_Chevalier1995}
\begin{eqnarray}
\label{MSE_WLMSE}
    \varepsilon_{WL}^2
    &\define&  \EX[\|\hat{\x}^{\text{WLMMSE}} - \x\|^2]
  \nonumber\\
    &=&
    \tr
    \bigg(
    \Gammamat_{\x\x} 
    - 
    \bigg[
    \begin{array}{cc}
         \hat\H_1  &  
         \hat\H_2 
    \end{array}
    \bigg] 
    {\Gammamat}_{\underline\y\underline\y}
    \bigg[
    \begin{array}{c}
         \hat\H_1^H  \\
         \hat\H_2^H 
    \end{array}
    \bigg]
    \bigg)
    ,
\end{eqnarray} 
where ${\Gammamat}_{\underline\y\underline\y}$ is the augmented covariance matrix from \eqref{augmented_cov}.

For the special case where $\y$ is a proper random vector (i.e. $\C_{\y\y} = \ve{0}$), and  $\y$ and $\x$ are jointly proper random vectors (i.e. $ \C_{\x\y} = {\C}_{\x\y}=\ve{0}$), it can be shown that \eqref{H1_optimal_WLMMSE}, \eqref{H2_optimal_WLMMSE}, and \eqref{Pyy} are reduced to
 $\hat{\H}_1=\Gammamat_{\x\y}\Gammamat_{\y\y}^{-1}$,  $\hat{\H}_2=\ve{0}$, and $\P_{\y\y} = \Gammamat_{\y\y}$, respectively.
By substituting these values in \eqref{WLMSE_est}, we obtain that, in this proper case, the WLMMSE estimator from \eqref{WLMSE_est} is reduced to  the  LMMSE estimator from \eqref{LMMSE}.
It is shown in \cite{Picinbono_Chevalier1995} that the difference between the \ac{mse} of the LMMSE estimator and the \ac{mse} of the WLMMSE estimator is given by 
\begin{equation} \label{eq:diff_MSE}
    \varepsilon_{L}^2 - \varepsilon_{WL}^2 = 
    \tr(\hat\H_2 \P_{\y\y}^* \hat\H_2^H),
\end{equation}
where $\varepsilon_{L}^2$ and $ \varepsilon_{WL}^2$ are defined in \eqref{MSE_LMMSE} and \eqref{MSE_WLMSE}, respectively.
It can be shown that the r.h.s. of \eqref{eq:diff_MSE} is always nonnegative as the Schur complement matrix, $\P_{\y\y}$, is a positive-definite matrix. Thus, there is always a gain in performing WLMMSE estimation over LMMSE estimation for improper random vectors \cite{Picinbono_Chevalier1995}. 
The LMMSE estimator 
can be interpreted as  ``one channel processing" operating on $\y$ only. The WLMMSE estimator is ``two channel processing" operating on $\y$ via $\H_1$ and, in parallel, on $\y_2$ via $\H_2$ \cite{JAHANCHAHI201433}.
In this paper, we consider ``two channel processing" in the graph frequency domain.

\subsection{Graph Signal Processing (GSP)}
\label{GSP_subsec}
Consider an undirected, connected, weighted graph $\mathcal{G}(\mathcal{V},\mathcal{E},\W)$, where $\mathcal{V}$ is the set of vertices and $\mathcal{E}$ is the set of edges, and the $N\times N$ real-valued matrix $\W$ is the nonnegative weighted adjacency matrix of the graph, where $N=|\mathcal{V}|$ is the number of vertices. The entry $\W_{i,j} \neq 0$ if there is an edge between node $i$ and node $j$; otherwise $\W_{i,j} = 0$. The Laplacian matrix of the graph is defined as
\begin{equation}
    \L = \diag(\W \ve{1}) - \W.
\end{equation}
Since the Laplacian matrix is symmetric and semi-positive definite then  its eigenvalue decomposition is given by
\begin{equation}
    \L = \V {\text{diag}}(\ve{\lambda}) \V^{-1},
\end{equation}
where $\ve\lambda
$ is a vector with the ordered eigenvalues of $\L$,  $0 = \lambda_1\leq \lambda_2\ldots\leq\lambda_N$, 
$\V$ is a unitary matrix, where its $n$th column, $\v_n$,  is the  \textit{real-valued} eigenvector
associated with the $n$th eigenvalue, $\lambda_n$. A graph signal is an $N$-dimensional vector $\y$ that assigns a scalar value to each vertex in the graph,  $y_n$, $n=1,2,\ldots,N$. The GFT of the graph signal $\y$ is defined as
\begin{equation}
\label{GFT_def}
    \bar\y = \V^T\y,
\end{equation}
and the inverse GFT is given by $\y = \V\bar\y$. 

Linear and shift-invariant graph filters with respect to graph shift operators (GSOs) generalize linear time-invariant filters used in time series. Let the GSO be the Laplacian matrix $\L$. A graph filter applied to the GSO is  a function $\f(\cdot)$, where
\begin{equation}
\label{graph_filter_def}
    \f(\L) = \V {\text{diag}}(\f(\ve\lambda))\V^T,
\end{equation}
in which $\f(\ve\lambda) = (f(\lambda_1),f(\lambda_2),\ldots,f(\lambda_N))^T$and $f(\lambda_n)\in{\mathbb{C}}$ is the  filter graph frequency response  at graph frequency $\lambda_n$. 


\begin{figure} 
    \begin{minipage}{.35\textwidth}
    \begin{tikzpicture}[node distance=2cm,thick] 
    \matrix[column sep = .95cm, row sep = .1cm]
    {   & & & & \\
    \node (a1){$\y$}; & \node [block](A2){$\V^T$}; & \node [block](A3){${\text{diag}}(\f(\ve\lambda))$}; &   \node [block](B8){$\V$}; & \node (b9){$\hat {\x}$};\\ 
    }; 
    \draw [->] (a1) to (A2);
    \draw [-] (A2) --node [above]{$\bar{\y}$}  (A3);
    \draw [-] (A3) --node [above]{$\hat{\bar{\x}}$} (B8) ; 
    \draw [->] (B8) to (b9) ;  
    \end{tikzpicture}
        \vspace{-0.5cm}
    \caption{GSP linear estimator \eqref{GSPlinear}}
    \label{fig:GLMMSE proper}
    \end{minipage}
    \begin{minipage}{.4\textwidth}
    \begin{tikzpicture}[auto, node distance=2cm,thick] 
    \matrix[column sep = .5cm, row sep = .1cm]
    {   & & & & & &  \\
    \node (a1){$\y$}; & \node [block](A2){$\V^T$}; & \node [block](A3){${\text{diag}}(\f_1(\ve\lambda))$}; &  \node [coordinate](a4){}; & &  &  \\
     & & &  \node [sum](add1){\suma}; &  & \node [block](B8){$\V$}; & \node (b9){$\hat {\x}$};\\
    \node (c1){$\y^*$}; & 	\node [block](C2){$\V^T$}; & \node [block](C3){${\text{diag}}(\f_2(\ve\lambda))$}; &   \node [coordinate](c4){};  & & &  \\ 
    }; 
    \draw [->] (a1) to (A2);
    \draw [-] (A2) --node [above]{$\bar{\y}$}  (A3);
    \draw [-] (A3) to (a4) ;
    \draw [->] (a4) to (add1); 
    \draw [->] (c1) to (C2);
    \draw [-] (C2) --node [above]{$\bar{\y}^*$} (C3);    
    \draw [-] (C3) to (c4);
    \draw [->] (c4) to (add1) ;
    \draw [-] (add1) --node [above]{$\hat{\bar{\x}}$}  (B8);
    \draw [->] (B8) to (b9) ; 
    \end{tikzpicture}
    \vspace{-0.5cm}
    \caption{Widely-linear GSP estimator \eqref{WLMMSE_general_form}}
    \label{fig:GLMMSE improper}
    \end{minipage}  
    \label{fig:LMMSE and WLMMSE for graph signals}
        \vspace{-0.5cm}
\end{figure}

\subsection{GSP-LMMSE Estimator}
\label{GSP_LMMSE_subsec}
The GSP-LMMSE estimator for graph signals was proposed in \cite{kroizer2021bayesian} for  real-valued random vectors,  and is based on diagonal covariance matrices. Similar to the derivations in \cite{kroizer2021bayesian}, it can be shown that for zero-mean complex-valued random vectors $\x$ and $\y$, the  GSP-LMMSE estimator is given by
 \begin{equation}
\label{opt_LMMSE_GSP1}
    \hat{\x}^{(\text{GSP-LMMSE})}  =   \Vmat \D_{\bar{\x}\bar{\y}}^\Gammamat(\D_{\bar{\y}\bar{\y}}^\Gammamat)^{-1}\bar\y,
\end{equation}
where
\begin{eqnarray} \label{d_def1}
	\Dmat_{\bar{\x}\bar{\y}}^\Gammamat \define  {\text{ddiag}}(\Gammamat_{\bar{\x}\bar{\y}}),~~~
	\Dmat_{\bar{\y}\bar{\y}}^\Gammamat \define 
	 \text{ddiag}(\Gammamat_{\bar{\y}\bar{\y}} ).
\end{eqnarray}
It should be noted that $\Dmat_{\bar{\y}\bar{\y}}^\Gammamat$ is a real diagonal matrix. Thus, in the following, we will use $(\Dmat_{\bar{\y}\bar{\y}}^\Gammamat)^H=\Dmat_{\bar{\y}\bar{\y}}^\Gammamat$.
\begin{Claim}
\label{new_claim}
The  GSP-LMMSE estimator in \eqref{opt_LMMSE_GSP1} minimizes the \ac{mse}, $\EX[\|\hat{\x} - \x\|^2]$, among linear estimators that are the output of a graph filter, i.e. estimators of the form
\be
\label{GSPlinear}
 \hat{\x}= \V {\text{diag}}(\f(\ve\lambda)) \V^T \y,
\ee
where $\f(\cdot)$ is a graph filter.
\end{Claim}
\begin{IEEEproof}
    The proof appears in Appendix \ref{appA_GSP_LMMSE}. The general form of the estimators in \eqref{GSPlinear} is presented in Fig. \ref{fig:GLMMSE proper}. It can be seen that for the case where $\x$ and $\y$ are real-valued signals, we obtain the real-value GSP-LMMSE estimators from \cite{kroizer2021bayesian}.
\end{IEEEproof}

The GSP-LMMSE estimator from \eqref{opt_LMMSE_GSP1} has various advantages that make it attractive for practical applications  with real-valued signals. 
First, it is based on the {\em{diagonal}} of the covariance matrices of the  signals in the graph frequency domain $\bar{\x}$ and $\bar{\y}$, $\D_{\bar{\x}\bar{\y}}^\Gammamat$ and $\D_{\bar{\y}\bar{\y}}^\Gammamat$, that are defined in \eqref{d_def1}.
Thus, it  has advantages from a computational point of view, compared with the LMMSE estimator in \eqref{LMMSE} that uses the full covariance matrices of $\x$ and $\y$. In particular, it requires inverting a diagonal matrix, $\D_{\bar{\y}\bar{\y}}^\Gammamat$, compared with the LMMSE that requires inverting a full $N\times N$ matrix.
While the computational complexity is lower, the simulations shown in \cite{kroizer2021bayesian} revealed that the MSE performance of the GSP-LMMSE estimator is close to  that of the LMMSE estimator.
Moreover, when the estimators are implemented via data-based learning of the second-order statistics, the GSP-LMMSE estimator can be implemented even with a limited training dataset, where the LMMSE estimator is unstable in this regime. Furthermore, it is shown in \cite{kroizer2021bayesian} that the GSP-LMMSE estimator can be easily updated when the topology of the graph changes without generating a new dataset. Finally, simulations in \cite{kroizer2021bayesian} show that the implementations of the GSP-LMMSE by parametric graph filters are significantly more robust to outliers and to network topology changes in the form of adding/removing vertices/edges.
However, the GSP-LMMSE estimator was developed for real-valued signals, and thus does not utilize  the theory of augmented complex statistics \cite{schreier2010statistical,mandic2009complex}.
In the following, we propose GSP-based widely-linear estimation.


\section{WLMMSE for Graph Signals}
\label{new_estimator_sec}
In this section, the concept of widely-linear estimation is  applied to the case of improper graph signals.  In Subsection \ref{GSP_WLMMSE_subsec}, we develop the GSP-WLMMSE estimator for improper graph signals.  In Subsection  \ref{discussion_subsec}, we discuss   the performance and the computational complexities of the estimators. The practical implementation of the estimator with data-based sample covariance matrices is described in Subsection \ref{sample_version}.

\subsection{GSP-WLMMSE Estimator}
\label{GSP_WLMMSE_subsec}
In this subsection, we develop the GSP-WLMMSE estimator for a general improper complex-valued random vector.
The idea of the proposed GSP-WLMMSE estimator is to generalize the  ``one-channel processing" of the GSP-LMMSE estimator operating on $\y$ from Subsection \ref{GSP_LMMSE_subsec} to  ``two channel processing" operating on both $\y$ and $\y^*$, where these two channels are in the form of an output of two graph filters.

To this end, we consider widely-linear GSP estimators of $\x$ based on both $\y$ and $\y^*$ that have the form (see Fig. \ref{fig:GLMMSE improper})
\begin{eqnarray}
\label{WLMMSE_general_form}
    \hat{\x} &=&  \V {\text{diag}}(\f_1(\ve\lambda)) \V^T \y + \V \diag(\f_2(\ve\lambda)) \V^T \y^* .
\end{eqnarray}
We also define the matrices
\begin{eqnarray} \label{d_def2}
	 	\Dmat_{\bar{\x}\bar{\y}}^\C \define  {\text{ddiag}}(\C_{\bar{\x}\bar{\y}}),~~~
	\Dmat_{\bar{\y}\bar{\y}}^\C \define 
	 \text{ddiag}(\C_{\bar{\y}\bar{\y}} ),
\end{eqnarray}
where $\bar\x$ and $\bar\y$ 
are the GFTs of $\x$ and $\y$, respectively.
It should be noted that the description of the diagonal covariance matrices in \eqref{d_def1} and \eqref{d_def2} is for the sake of notation simplicity. In practice, there is no need to compute the full covariance matrices and then take the diagonal, and only the diagonal elements should be computed. For example, in order to compute $\Dmat_{\bar{\x}\bar{\y}}^\C$, one only needs to compute $N$ elements: $\EX[\bar{x}_n\bar{y}_n]$, $n=1,\ldots,N$.
We also define 
\begin{equation}
\label{rho_def}
    \rho_n \define \frac{| [\C_{\bar\y\bar\y}]_{n,n}|^2}{| [\Gammamat_{\bar\y\bar\y}]_{n,n}|^2}=\frac{| [\D_{\bar\y\bar\y}^\C]_{n,n}|^2}{| [\D_{\bar\y\bar\y}^\Gammamat]_{n,n}|^2}, ~n=1,\ldots,N.
    \end{equation}
It can be seen that the parameter $\rho_n $ is a measure of the impropriety of the random variable $\bar{y}_n$. In particular, if $\bar{\y}$ is a proper vector then $\rho_n=0$, $n=1,\ldots,N$.

The following theorem describes the GSP-WLMMSE estimator for a general improper model.
\begin{Theorem} \label{Theorem1}
Let us assume that $\C_{\y\y}\neq \zerovec$ and $  \rho_n \neq 1$, $n=1,\ldots,N$.
Then,
    the GSP-WLMMSE estimator, which is the estimator that achieves the minimum \ac{mse} within the class of widely-linear GSP estimators with the form  \eqref{WLMMSE_general_form}, for zero-mean complex graph signals,  is  given by 
\beqna
\label{opt_WLMMSE_GSP}
    \hat{\x}^{(\text{GSP-WLMMSE})}  \hspace{5.75cm}
    \nonumber\\
    =
     \V   (\R\Dmat_{\bar{\x}\bar{\y}}^\Gammamat (\Dmat_{\bar{\y}\bar{\y}}^\Gammamat)^{-1}
    +(\I-\R){\Dmat}_{\bar{\y}\bar{\x}}^\C ({\Dmat}_{\bar{\y}\bar{\y}}^\C)^{-1}
  )\V^T \y \hspace{0.75cm}\nonumber\\+ \V (\R{\Dmat}_{\bar{\y}\bar{\x}}^\C(\Dmat_{\bar{\y}\bar{\y}}^\Gammamat)^{-1}
    +(\I-\R){\Dmat}_{\bar{\x}\bar{\y}}^\Gammamat(({\Dmat}_{\bar{\y}\bar{\y}}^\C)^{-1})^*
  ) \V^T \y^*,
\eeqna
where the diagonal coefficient matrix is
\begin{equation}
\label{R1_def}
\R\define {\text{diag}}\left(\left[\frac{1}{1 - \rho_1},\frac{1}{1 - \rho_2},\ldots,\frac{1}{1 - \rho_N}\right]\right).
\end{equation}
\end{Theorem}
\begin{IEEEproof}
    The proof appears in Appendix \ref{appA}.
\end{IEEEproof}
It is also shown in Appendix \ref{appA} that under the conditions of  Theorem \ref{Theorem1}, except for the condition that $\C_{\y\y}\neq \zerovec$,
the GSP-WLMMSE estimator from \eqref{opt_WLMMSE_GSP} can be written as
\begin{eqnarray}
\label{opt_WLMMSE_GSP2}
\begin{aligned}
    \hat{\x}^{(\text{GSP-WLMMSE})}  \hspace{5.75cm}
   \\
    =
    \V(
    \D_{\bar\x\bar\y}^{\Gammamat}  - 
   \D_{\bar\x\bar\y}^{\C}
     (\D_{\bar\y\bar\y}^{\Gammamat})^{-1} (\D_{\bar\y\bar\y}^{\C} )^* )
     ( \D_{\bar\y\bar\y}^\P )^{-1}\V^T \y \\
    + 
    \V
    (
    (\D_{\bar\y\bar\x}^{\C}
    -
    \D_{\bar\y\bar\y}^{\C}(\D_{\bar\y\bar\y}^{\Gammamat})^{-*}(\D_{\bar\y\bar\x}^{\Gammamat} )^*
    )
    ( \D_{\bar\y\bar\y}^\P )^{-1}
    \V^T \y^*,
\end{aligned}    
\end{eqnarray}
where the Schur complement $\D_{\bar\y\bar\y}^\P$ is defined 
as
\begin{equation}
\label{Schur_diag}
    \D_{\bar\y\bar\y}^\P\define
    \D_{\bar\y\bar\y}^{\Gammamat}  -   \D_{\bar\y\bar\y}^{\C} (\D_{\bar\y\bar\y}^{\Gammamat})^{-*}(\D_{\bar\y\bar\y}^{\C})^*.
\end{equation}
\begin{remark}
The Schur complement matrix $\D_{\bar\y\bar\y}^\P$ is a diagonal matrix with non-negative values on its diagonal.
Since we assume in this paper that $\Gammamat_{\y\y}$ is a non-singular matrix, 
$\D_{\bar\y\bar\y}^\P$ can also be written as 
\be
\label{ppp}
\D_{\bar\y\bar\y}^\P = \D_{\bar\y\bar\y}^{\Gammamat}(\I - (\D_{\bar\y\bar\y}^{\Gammamat})^{-1}\D_{\bar\y\bar\y}^{\C} (\D_{\bar\y\bar\y}^{\Gammamat})^{-*}(\D_{\bar\y\bar\y}^{\C})^*).
\ee
Observe that $(\D_{\bar\y\bar\y}^{\Gammamat})^{-1}\D_{\bar\y\bar\y}^{\C} (\D_{\bar\y\bar\y}^{\Gammamat})^{-*}(\D_{\bar\y\bar\y}^{\C})^*$ is a diagonal matrix with the $(n,n)$th entry on its main diagonal given as $\rho_n$, $n=1,2,\ldots,N$. Hence, the Schur complement matrix can be expressed as a product of diagonal matrices: $\D_{\bar\y\bar\y}^\P = \D_{\bar\y\bar\y}^{\Gammamat}\R^{-1}$.
\end{remark}
It is also shown in Appendix \ref{appA} that the optimal graph filters, i.e. the graph filters 
in \eqref{WLMMSE_general_form} that will lead to  the GSP-WLMMSE estimator from
\eqref{opt_WLMMSE_GSP2}, $\hat{\x}^{(\text{GSP-WLMMSE})} $, 
are
\begin{eqnarray}
\label{f1_paper}
    \hat{\f}_1 (\ve\lambda)=
    ( \D_{\bar\y\bar\y}^\P )^{-1}
    (
    (\d_{\bar\y\bar\x}^{\Gammamat} )^* - 
    (\D_{\bar\y\bar\y}^{\C} )^* (\D_{\bar\y\bar\y}^{\Gammamat})^{-1} \d_{\bar\y\bar\x}^{\C}
    )
   \eeqna
   and
   \beqna
   \label{f2_paper}
    \hat{\f}_2 (\ve\lambda)=
    ( \D_{\bar\y\bar\y}^\P )^{-1} 
    (\d_{\bar\y\bar\x}^{\C}
    -
    \D_{\bar\y\bar\y}^{\C}(\D_{\bar\y\bar\y}^{\Gammamat})^{-1}(\d_{\bar\y\bar\x}^{\Gammamat} )^*
    ),
\end{eqnarray}
where $\d_{\bar\y\bar\x}^{\C}\define
{\text{diag}}(\C_{\bar{\y}\bar{\x}})
$ and $\d_{\bar\y\bar\x}^{\C} \define
{\text{diag}}(\C_{\bar{\y}\bar{\x}})
$, and $  \D_{\bar\y\bar\y}^{\Gammamat}$ and  $\D_{\bar\y\bar\y}^{\C} $
are defined in \eqref{d_def1} and \eqref{d_def2}, respectively. 
That is, we can rewrite the GSP-WLMMSE from \eqref{opt_WLMMSE_GSP} as 
\beqna
\label{filtered_opt_WLMMSE_GSP}
    \hat{\x}^{(\text{GSP-WLMMSE})}\hspace{5cm}\nonumber\\=
    \V {\text{diag}}(\hat\f_1(\ve\lambda)) \V^T \y + \V \diag(\hat\f_2(\ve\lambda)) \V^T \y^*.
    \eeqna

The advantages of the form in \eqref{opt_WLMMSE_GSP2} compared with the form in \eqref{opt_WLMMSE_GSP} is that it also holds for $\C_{\y\y}= \zerovec$, and that its form is reminiscent of the classical WLMMSE estimator in \eqref{WLMSE_est}. The form in \eqref{opt_WLMMSE_GSP} enables tractable analysis of special cases and demonstrates the influence of  the impropriety parameter, $\rho_n $.

\subsection{Discussion: Performance and Complexity}
\label{discussion_subsec}
\subsubsection{Order relation between the estimators}
\label{order_relation_subsection}
The class of widely-linear estimators is described in \eqref{general_form_widely_lin}. By setting $\H_2=\zerovec$ in \eqref{general_form_widely_lin}, we obtain the class of linear estimators. 
If we set $\H_2=\zerovec$ and restrict $\H_1$ to $\H_1=\V {\text{diag}}(\f(\ve\lambda)) \V^T$ in \eqref{general_form_widely_lin}, we obtain the general GSP-linear estimators in
\eqref{GSPlinear}.
Thus, the GSP-linear estimators are included in the linear estimators, which are included as a subset of widely-linear estimators. 
Similarly,  the set of widely-linear GSP estimators in \eqref{WLMMSE_general_form} includes the set of GSP-linear estimators in \eqref{GSPlinear}, since
\eqref{GSPlinear} can be obtained from \eqref{WLMMSE_general_form} by setting 
$\f_2(\ve\lambda)=\zerovec$. On the other hand, the class of widely-linear estimators from  \eqref{general_form_widely_lin} includes 
the set of widely-linear GSP estimators in \eqref{WLMMSE_general_form}, since by setting 
 $\H_1=\V {\text{diag}}(\f_1(\ve\lambda)) \V^T $ and $\H_2 =\V \diag(\f_2(\ve\lambda)) \V^T$  in \eqref{general_form_widely_lin} we obtain the class in \eqref{WLMMSE_general_form} (i.e. \eqref{WLMMSE_general_form} is more restrictive). 
 Finally, we note that there are widely-linear GSP estimators that are not linear (e.g., the estimator in \eqref{new_eq}), while there are linear estimators that are not widely-linear GSP estimators (any estimator in the form of \eqref{LMMSE} as long as  $\Gammamat_{\x\y}\Gammamat_{\y\y}^{-1}$ is not a graph filter). These relationships are described  in Fig. \ref{diagram}.

\subsubsection{Performance}
Since widely-linear GSP estimators in \eqref{WLMMSE_general_form} include the GSP-LMMSE estimator in \eqref{opt_LMMSE_GSP1} as a special case, and since the GSP-WLMMSE estimator has the lowest  \ac{mse} among all estimators of the form in \eqref{WLMMSE_general_form}, the \ac{mse} of the GSP-WLMMSE estimator is lower than or equal to the \ac{mse} of the GSP-LMMSE estimator.
The following theorem gives a closed-form expression for the MSE gap  of these estimators.
\begin{Theorem}\label{Theorem5}
Let the \ac{mse}s of the GSP-LMMSE and the GSP-WLMMSE estimators be denoted by 
$\varepsilon_{GSP-L}^2$ and $\varepsilon_{GSP-WL}^2$.
Then, the difference between the MSEs is 
\begin{equation}\label{eq:diff_mse}   
    \varepsilon_{GSP-L}^2-\varepsilon_{GSP-WL}^2 =
    \hat{\f}_2^H (\ve\lambda)
    \D_{\bar\y\bar\y}^\P
    \hat{\f}_2 (\ve\lambda),
\end{equation}
where $ \D_{\bar\y\bar\y}^\P $ and $\hat{\f}_2(\ve\lambda)$ are  given in \eqref{Schur_diag} and
    \eqref{f2_paper}, respectively.
\end{Theorem}
\begin{IEEEproof}
The proof appears in Appendix \ref{appB}.
\end{IEEEproof}

Since the Schur complement, $\D_{\bar\y\bar\y}^\P$, is a positive-definite matrix,  the quadratic form in \eqref{eq:diff_mse}  is always nonnegative: $\varepsilon_{GSP-L}^2-\varepsilon_{GSP-WL}^2\geq 0$. 
Thus, in the case of improper graph signals, there is a performance gain in GSP-WLMMSE estimation over GSP-LMMSE estimation. Moreover, it can be seen that this gain is only a function of the graph filter $\hat{\f}_2(\ve\lambda)$, which filters the complex conjugate part of the observation, $\y^*$, that has relevant information in the improper case. 
It should be noted that since the GSP-WLMMSE belongs to the family of  widely-linear estimators, it achieves intermediate performance between the
optimal widely-linear filter (WLMMSE estimator) and the optimal linear filter in the graph domain (GSP-LMMSE estimator), as discussed in Subsection \ref{order_relation_subsection}.

\subsubsection{Orthogonality principle}
The WLMMSE estimator  satisfies  the orthogonality principle \cite{Picinbono_Chevalier1995}, which states that \[(\hat\x^{(\text{WLMMSE})}-\x) \perp \y{\text{~~~and~~~}}(\hat\x^{(\text{WLMMSE})}-\x) \perp \y^*.\]
Similarly, in Appendix \ref{app_calc} it is shown that 
\[(\hat\x^{(\text{GSP-WLMMSE})}-\x) \perp \y{\text{~~~and~~~}} (\hat\x^{(\text{GSP-WLMMSE})}-\x) \perp \y^*. \]
 In other words, the estimation error vector of the GSP-WLMMSE estimator  is orthogonal to both $\y$ and $\y^*$.

\subsubsection{Computational complexity}
In terms of computational complexity, 
the  WLMMSE estimator from \eqref{WLMSE_est} requires: 
1) computing the inverse of the $N\times N$ complex-valued matrices ${\ve\Gamma}_{\y\y}$ and $ \P_{\y\y}$, which has a complexity of $\mathcal{O}(N^3)$; and  2) performing full matrix-vector multiplications, with a  computational complexity of $\mathcal{O}(N^2)$.
The GSP-WLMMSE estimator from \eqref{opt_WLMMSE_GSP} requires:
1)
computing the inverse of the diagonal matrices ${\Dmat}_{\bar{\y}\bar{\y}}^\Gamma$ and ${\Dmat}_{\bar{\y}\bar{\y}}^\C$, which has a complexity of $\mathcal{O}(N)$;  and 2) performing multiplications of diagonal matrices with a cost of $\mathcal{O}(N)$. Then, to obtain the estimate in the vertex domain, it is required to multiply the previous computations by $\V$ and  $\V^T\y$. These matrix-vector multiplications have  a cost of $\mathcal{O}(N^2)$.
To conclude, 
if $\V$ is given,
the final computational complexity of the WLMMSE estimator is $\mathcal{O}(N^3)$ due to the need to compute the inverse of a full matrix, while the GSP-WLMMSE has a complexity of $\mathcal{O}(N^2)$.

The use of the graph frequency domain in the GSP-WLMMSE estimator requires the computation of the eigenvalue decomposition (EVD) of the Laplacian matrix, which is of order $\mathcal{O}(N^3)$. If the EVD can be assumed to be known, then this task may be avoided. In addition, since the EVD is only a function of the topology of the graph and is independent of the data, 
the eigendecomposition of the
graph Laplacian matrix can be performed offline and there is no need for a recalculation with new data.
Recent works propose low-complexity  methods to reduce the complexity of this task (see e.g., \cite{SVD}). In addition,
the computational complexity of the GSP-LMMSE estimator can be reduced even further 
by using the properties of 
the Laplacian matrix, which tends to be sparse, and therefore,
matrix operations may require fewer computations.
 One can also accelerate
the computation of the GSP-WLMMSE estimator by using techniques to approximate the graph filters, as described in the following subsection.

\subsubsection{Chebyshev Polynomial Approximation}
\label{Chebyshev_subsec}
Computing the full EVD of the graph Laplacian matrix, as needed for the GSP-WLMMSE estimator, becomes
computationally challenging as  $N$, increases.
A common
way to avoid this computational burden is to use polynomial approximations.
One such approach is the truncated Chebyshev polynomial
approximation of a graph filter \cite{Chebyshev_2017,Chebyshev_2018,stankovic2019graph,Shuman_2020}.
In addition to reduced computational complexity, 
the second benefit of the Chebyshev polynomial approximations is that their recurrence relations make them readily amenable for distributed computation.

In this case, instead of exactly computing   the  GSP-WLMMSE estimator from \eqref{filtered_opt_WLMMSE_GSP}, which requires explicit computation of the
entire set of eigenvectors  of $\Lmat$, we approximate  $\hat\f_1(\L)=\V {\text{diag}}(\hat\f_1(\ve\lambda)) \V^T$ and $\hat\f_2(\L)=\V {\text{diag}}(\hat\f_2(\ve\lambda)) \V^T$ by the polynomial approximations $\g_1(\Lmat) $ and $\g_2(\Lmat) $, respectively, where    $\hat{\f}_1 (\ve\lambda)$ and $\hat{\f}_2 (\ve\lambda)$
 are defined in \eqref{f1_paper}  and \eqref{f2_paper}. The terms
$g_1(\Lmat)$ and $\g_2(\Lmat)$ are   computed by truncating a shifted Chebyshev series expansion of the functions $ \hat\f_1(\ve\lambda)$ and $\hat\f_2(\ve\lambda)$, as described in \cite{Chebyshev_2017,Chebyshev_2018,Shuman_2020,stankovic2019graph}.
 Then, the estimator in \eqref{filtered_opt_WLMMSE_GSP} is replaced by the approximation
 \beqna
\label{filtered_opt_WLMMSE_GSP_cheby}
    \hat{\x}^{(\text{GSP-WLMMSE})} \approx 
   g_1(\Lmat) \y + g_2(\Lmat) \y^*.
    \eeqna
The  graph signal filtering in \eqref{filtered_opt_WLMMSE_GSP_cheby} is implemented in the vertex
domain in a local manner, since the Chebyshev polynomials
are  readily
amenable to distributed computation \cite{Chebyshev_2018}.

 The computational complexity of computing  the Chebyshev polynomial  approximation  in \eqref{filtered_opt_WLMMSE_GSP_cheby} is $\mathcal{O}(L|\mathcal{E}|)$,  where $|\mathcal{E}|$ is the number of edges and  $L$ is the polynomial order.
 For large sparse graphs, this  is approximately linear in  $N$  \cite{Chebyshev_2018,Shuman_2020}.
This is  as opposed to $\mathcal{O}(N^3)$ required to naively compute the full set of eigenvectors  of $\Lmat$,  which is necessary for the GSP-WLMMSE estimator. 
Moreover, the estimator in \eqref{filtered_opt_WLMMSE_GSP_cheby}
 can be performed in a distributed setting where each vertex knows only its own signal value and can communicate solely with its neighboring vertices.
The amount of communication required to perform the distributed computations scales with the size of the network  through the number of graph edges. Therefore, the method is well suited to large-scale sparse networks.

\subsubsection{General design of the graph frequency response}
\label{General_design_sub}
In Theorem \ref{Theorem1} we presented a general approach to finding the optimal GSP-WLMMSE estimator. However, the GSP-WLMMSE estimator is a function of the specific graph structure with fixed dimensions. 
In particular, the optimal graph frequency responses 
of the GSP-WLMMSE estimator from \eqref{f1_paper} and \eqref{f2_paper} are defined only at graph frequencies $\lambda_n$, $n=1,\dots,N$. 
Thus, the optimal graph frequency response needs to be redeveloped for any small change in the topology. Alternatively, we can formulate the problem of designing  widely-linear estimators of the form in \eqref{WLMMSE_general_form} that minimizes the \ac{mse}, but where $\f_1(\ve\lambda)$ and $\f_2(\ve\lambda)$ are restricted to specific parametrizations as
two graph filters, $h_1(\cdot;\alphavec_1)$ and $h_2(\cdot;\alphavec_2)$, where $\alphavec_1$ and  $\alphavec_2$ contain the graph filters' parameters. These graph filters can be, for example, 
 the 
auto-regressive moving-average (ARMA)  graph filter \cite{Isufi_Leus2017}  with the following  entries on the diagonal:  
\begin{equation}
\label{ARMA2}
    h_i(\lambda_n;\alphavec) = \frac{\sum_{p=0}^{P-1} \lambda_n^p c_p}{1 + \sum_{q=1}^{Q-1} \lambda_n^q a_q},
\end{equation}
where $\alphavec=[c_0,\ldots,c_{P-1},a_1,\ldots,a_Q]^T$  and $P,Q$ are the filter orders. 
The goal is to find the graph filters 
that generate MSE-optimal widely-linear estimator with specific parametrization.

The MSE-optimal parameter vectors, $\alphavec_1$ and $\alphavec_2$, for any graph-filter parametrizations, $h_i(\ve\lambda;\alphavec)$, $i=1,2$, are found by minimizing the MSE after substituting in \eqref{min_f1_f2}
 from Appendix \ref{appA} the specific parametrizations:
\begin{eqnarray}
\label{MSEh1h2}
\EX[\|\hat{\x} - \x\|^2]\hspace{6.25cm}\nonumber\\=
     \EX[\|{\text{diag}}(\h_1(\ve\lambda;\alphavec_1)) \bar\y + {\text{diag}}(\h_2(\ve\lambda;\alphavec_2)) \bar{\y}^*  - \bar\x\|^2].
\end{eqnarray}

This universal design of graph filters by finding the frequency response over a continuous range of graph frequencies is adaptive to topology changes, e.g., in cases when the number of vertices and/or edges is changed. In addition, the parametric representation provides a straightforward way to integrate GSP properties. Finally, since in practice the desired graph frequency response is often approximated by its sample-mean version (see  Subsection \ref{sample_version}), parametrizations can reduce outlier errors and noise effects. More technical details on this approach can be found in \cite{kroizer2021bayesian} for the GSP-LMMSE estimator.

\subsubsection{Alternative derivation by real and imaginary parts}
\label{real_subsec}
An alternative approach to deal with complex-valued signals is to 
minimize the MSE  separately w.r.t. the real and imaginary parts of $\x$,  $\x_\Re=\Re\{\x\}\in \mathbb{R}^N$ and $\x_\Im=\Im\{\x\}\in \mathbb{R}^N$.
Then the goal is to estimate the  vector $[\x_\Re^T,\x_\Im^T]^T\in \mathbb{R}^{2N}$ from the observation vector $[\y_\Re^T,\y_\Im^T]^T=[\Re\{\y\}^T,\Im\{\y\}^T]^T\in \mathbb{R}^{2N}$. Then, instead of the  widely-linear GSP estimators considered in \eqref{WLMMSE_general_form},
we consider  a real-valued GSP estimator of the form:
 	\begin{eqnarray}\label{WLMMSE_general_form_real_value}  
        \left[ 
            \begin{array}{c}
                 \hat\x_\Re \\
                 \hat\x_\Im 
            \end{array}
        \right]\hspace{6.25cm}
        \nonumber\\ =
        \left[ \hspace{-0.1cm}
            \begin{array}{cc}
                \V {\text{diag}}(\g_{11}(\ve\lambda))\V^T  &    \V {\text{diag}}(\g_{12}(\ve\lambda))\V^T\\
                   \V {\text{diag}}(\g_{21}(\ve\lambda))\V^T  &    \V {\text{diag}}(\g_{22}(\ve\lambda))\V^T 
        \end{array}\hspace{-0.1cm}
        \right]
        \left[ \hspace{-0.1cm}
            \begin{array}{c}
                 \y_\Re \\
                 \y_\Im 
            \end{array}\hspace{-0.1cm}
        \right],   
    \end{eqnarray}
where $\g_{ij}(\cdot)$, $i=1,2$, are four {\em{real-valued}} graph filters. 

In this representation, we minimize the MSE,
\[
\EX[\|\hat{\x} - \x\|^2]=\EX[(\hat\x_\Re - \x_\Re)^2]
+\EX[(\hat\x_\Im - \x_\Im)^2],
\]
w.r.t. the four {\em{real-valued}} graph filters. Similarly to Appendix \ref{appA}, it can be shown that this minimization results in 
\begin{eqnarray}
\label{optimal_real_filters}
&&\hat{\g}_{kk}(\ve\lambda)=\Re\{\hat\f_{1}(\ve\lambda)\}+(3-2k)\Re\{\hat\f_{2}(\ve\lambda)\},\nonumber\\&&
 \hat\g_{km}(\ve\lambda)=(k-m)\Im\{\f_{1}(\ve\lambda)\}+\Im\{\f_{2}(\ve\lambda)\}, 
 \end{eqnarray}
 $k,m=1,2$, $k\neq m$,
 where $\hat{\f}_1 (\ve\lambda)$ and $\hat{\f}_2 (\ve\lambda)$ are  the {\em{complex-valued}} optimal graph filters defined in \eqref{f1_paper} and \eqref{f2_paper}.
Thus, estimation based on the filtering of a signal and its complex conjugate by the two optimal complex-valued graph filters, i.e. the GSP-WLMMSE estimator from \eqref{opt_WLMMSE_GSP}, is equivalent, through a linear transformation, to the filtering of the real and imaginary parts of the signal by four real-valued graph filters, i.e. to  the estimator obtained by substituting \eqref{optimal_real_filters} in \eqref{WLMMSE_general_form_real_value}.

However, while the two approaches are mathematically equivalent, the latter approach is mathematically more cumbersome and less elegant and compact for the analysis of the improper complex-valued signals considered here. 
For example, the derivatives in the MSE optimization of the widely-linear approach are more tractable and less tedious, and require fewer additional assumptions \cite{Adali_Schreier_Scharf2011}.
In addition,  it is more convenient to develop the properties of the GSP-WLMMSE, as discussed in Section \ref{special_cases_subsec}.
Finally, the relations between the structure of the GSP-WLMMSE estimator from \eqref{opt_WLMMSE_GSP} and the real-valued GSP-LMMSE estimator are more straightforward  than for the real-valued version. 
Thus, the augmented complex representation used here has been favored by many researchers  (e.g., \cite{schreier2010statistical,mandic2009complex,Picinbono_Chevalier1995})  for processing improper signals in different tasks.

\color{black}
\subsection{Implementation with Sample Covariance Values}
\label{sample_version}
The recovery of random graph signals  by the discussed linear and widely-linear estimators can be used when the second-order statistics are completely specified.
However, for general nonlinear models, the  covariance matrices of the  graph signal and the observations 
do not have closed-form expressions.
In these cases, 
we can use data to estimate the second-order statistics  and then  recover the graph signals based on the estimated statistics. 
That is, assume that we have data (we can always generate such data for a given model) 
that is comprised of a training set consisting of $\Ntraining$ pairs that have the same joint distribution as $\x$ and $\y$, denoted by $\Data = \{\x_k, \y_k\}_{k=1}^{\Ntraining}$.
Then, the sample 
cross-covariance  and the sample complementary cross-covariance  matrices  are
\beqna
\label{one_sample}
\hat{\Gammamat}_{\x\y}
= \frac{1}{\Ntraining} \sum\nolimits_{k=1}^{\Ntraining} \x_k \y_k^H
{\text{ 
and }}
\hat{\C}_{\x\y} = \frac{1}{\Ntraining} \sum\nolimits_{k=1}^{\Ntraining} \x_k \y_k^T.
\eeqna
Similarly, we can calculate all the other sample diagonal and non-diagonal covariance matrices that are needed for the estimators.
Since the graph Laplacian is assumed to be known, in order to compute the sample 
cross-covariance  and the sample complementary cross-covariance  matrices {\em{in the graph frequency domain}}  one should take the data $\Data = \{\x_k, \y_k\}_{k=1}^{\Ntraining}$ and transform it according to \eqref{GFT_def}
to the graph frequency domain by $\bar\x_k\define \V^T\x_k$ and $\bar\y_k \define \V^T\y_k$, for any $k=1,\ldots,\Ntraining$.
Then, the sample covariance matrices used in the GSP-LMMSE, and GSP-WLMMSE estimators are, respectively,
\beqna
\label{one_sample2}
\hat{\Gammamat}_{\bar\x\bar\y}
= \frac{1}{\Ntraining} \sum\nolimits_{k=1}^{\Ntraining} \bar\x_k \bar\y_k^H
{\text{ 
and }}
\hat{\C}_{\bar\x\bar\y} = \frac{1}{\Ntraining} \sum\nolimits_{k=1}^{\Ntraining} \bar\x_k \bar\y_k^T.
\eeqna
It should be noted that 
for the GSP-WLMMSE,
the elements of the $N\times N$ diagonal sample covariance matrices  can be computed efficiently by
\beqna
\label{one_sample3}
[\hat{\D}_{\bar\x\bar\y}^{\Gammamat}]_{n,n}
= \frac{1}{\Ntraining} \sum_{k=1}^{\Ntraining} [\bar\x_k]_n [\bar\y_k^*]_n,~
[\hat{\D}_{\bar\x\bar\y}^{\C}]_{n,n}
= \frac{1}{\Ntraining} \sum_{k=1}^{\Ntraining} [\bar\x_k]_n [\bar\y_k]_n,\nonumber
\eeqna
$n=1,\ldots,N$.
Finally, the sample versions of the LMMSE, WLMMSE, GSP-LMMSE, and GSP-WLMMSE estimators are obtained
by plugging the sample covariance matrices  into  the original  estimators (see e.g., regarding the sample LMMSE estimator in \cite{6827237} and p. 728 in \cite{van2004optimum}, and regarding the sample GSP-LMMSE estimator in \cite{kroizer2021bayesian}). 
This approach can be interpreted as a
data-driven graph filter design with partial domain knowledge
by discriminative learning \cite{shlezinger2022discriminative}, where we use the available domain knowledge to choose the structure of the graph filters for the estimator, and then use data to obtain the matrices that minimize the MSE of the estimator ~\cite{shlezinger2022discriminative}. 

In many practical applications, the graph signal has a broad correlation function, so that estimating the covariance matrices from data with high accuracy necessitates a larger sample size than is available. In particular, in settings where the sample size  of the training data used to compute the sample-mean versions ($\Ntraining$) is comparable to the observation dimension ($N$), the  sample-LMMSE  and sample-WLMMSE estimators exhibit severe performance degradation  \cite{4914845,5484583,Lancewicki_Aladjem}. 
This is because the sample covariance matrix is not well-conditioned in the small sample size regime, and inverting it amplifies the estimation error.
In contrast, the GSP estimators, which rely on the inverse of a {\em{diagonal}} sample covariance matrix  (that can be computed for any $K\geq 1$ with probability 1) and exploit the graphical structure of the problem perform better in this regime. 
This makes the GSP estimators attractive for applications with large networks.
Moreover, the sample-GSP-LMMSE and the  sample-GSP-WLMMSE  estimators require the estimation of only two and four (respectively)   $N$-length vectors that contain the diagonal of the covariance matrices from \eqref{d_def1} and \eqref{d_def2}, while the sample-LMMSE and  sample-WLMMSE estimators require estimating two and four (respectively) $N\times N$  matrices.

This advantage improves the sample-GSP-WLMMSE estimator performance, compared to the sample-WLMMSE estimator, with  limited datasets used for the non-parametric estimation of the different sample-covariance matrices.
Moreover, the estimation of the inverse sample diagonal covariance matrices, $(\Dmat_{\bar{\y}\bar{\y}}^\Gammamat)^{-1}$ and
   $({\Dmat}_{\bar{\y}\bar{\y}}^\C)^{-1}$ from \eqref{opt_WLMMSE_GSP}, is more robust to limited data  than the estimation of the full sample covariance matrices inverse, required by the sample-WLMMSE estimator.

In terms of computational complexity, the sample-WLMMSE estimator requires: 
1) forming the sample mean 
and  the sample covariance matrices 
with full matrix multiplications and an additional cost of $\mathcal{O}(\Ntraining N^2)$; 2)
computing the inverse of the $N\times N$ 
 sample complex-valued matrices $\hat{\Gamma}_{\y\y}$ and $ \hat\P_{\y\y}$, which has a complexity of $\mathcal{O}(N^3)$; and  3) performing  matrix-vector multiplications, with a  computational complexity of $\mathcal{O}(N^2)$.
The sample-GSP-WLMMSE estimator  requires: 
1) forming the sample mean 
and  the {\em{diagonal}} sample covariance matrices 
with a cost of $\mathcal{O}(\Ntraining N)$, where the data is generated in the graph frequency domain;
2)
computing the inverse of the diagonal sample-covariance matrices, $\hat{\Dmat}_{\bar{\y}\bar{\y}}^\Gamma$ and $\hat{\Dmat}_{\bar{\y}\bar{\y}}^\C$, which has a complexity of $\mathcal{O}(N)$;  and 3) performing multiplications of diagonal matrices and matrix-vector multiplications  with a cost of $\mathcal{O}(N^2)$.
To conclude, 
if $\V$ is given,
the final computational complexity of the sample-WLMMSE estimator is $\mathcal{O}(N^2(N+\Ntraining))$, while the sample-GSP-WLMMSE has a complexity of $\mathcal{O}(N(N+\Ntraining))$.

\section{Special Cases and relation with existing estimators}
\label{special_cases_subsec}
In this section, we investigate the properties of the proposed GSP-WLMMSE estimator for the important special cases of proper signal and measurements (Subsection \ref{proper_sub}),
proper measurements with an improper signal (Subsection \ref{semi_proper_sub}), maximal improper observation vectors (Subsection \ref{max_improper_sub}), and real-valued unknown graph signals (Subsection \ref{real_sub}). In addition, we discuss the 
relations between 
the GSP-WLMMSE estimator 
and
 the scalar widely-linear estimators (Subsection 
\ref{scalar_sub}), and  the WLMMSE estimator (Subsection
\ref{relation_wlmse_sub}).

\subsection{Proper Case}
\label{proper_sub}
In the well-known proper case (see e.g., p. 35 in \cite{schreier2010statistical})
\renewcommand{\theenumi}{C.\arabic{enumi}}
\begin{enumerate}
	\setcounter{enumi}{0}
    \item
The signal and measurement are cross-proper:  $\C_{{\y}{\x}}=\ve{0}$.
\item
The measurement vector $\y$  is proper:
 $ \C_{\y\y} = \ve{0}$.
\end{enumerate}
Condition C.1 implies that
${\Dmat}_{\bar{\y}\bar{\x}}^\C=\zerovec$.

By substituting 
Condition C.2  in  \eqref{rho_def}-\eqref{R1_def}, we obtain that in this proper case
  $\rho_n=0$, $\forall n=1,\ldots,N$, and $\R=\Imat$. By substituting 
$\R=\Imat$ and ${\Dmat}_{\bar{\y}\bar{\x}}^\C=\zerovec$ in \eqref{opt_WLMMSE_GSP}, one obtains that in the proper case, the GSP-WLMMSE estimator is reduced to 
 the GSP-LMMSE estimator from \eqref{opt_LMMSE_GSP1}.

 \subsection{Proper Measurement Vector}
 \label{semi_proper_sub}
 In contrast to the previous case, if  only Condition C.2 holds, i.e. $ \C_{\y\y} = \ve{0}$,
 and  no specific assumption is introduced for
the estimand $\x$, then \eqref{opt_WLMMSE_GSP} is reduced to
\beqna
\label{opt_WLMMSE_GSP_C0}
    \hat{\x}^{(\text{GSP-WLMMSE})}  
    = \hat{\x}^{(\text{GSP-LMMSE})}  + \V {\Dmat}_{\bar{\y}\bar{\x}}^\C(\Dmat_{\bar{\y}\bar{\y}}^\Gammamat)^{-1}
    \bar{\y}^*,
\eeqna
where $ \hat{\x}^{(\text{GSP-LMMSE})} $ is given in \eqref{opt_LMMSE_GSP1}.
Thus, in this case the GSP-WLMMSE estimator is different from the GSP-LMMSE estimator.
In addition, by substituting 
$ \C_{\y\y} = \ve{0}$ and $\D_{\bar\y\bar\y}^{\C}=\ve{0}$ in \eqref{ppp} and \eqref{f2_paper}, we obtain that for this case $
\D_{\bar\y\bar\y}^\P = \D_{\bar\y\bar\y}^{\Gammamat}$ and
$
    \hat{\f}_2 (\ve\lambda)=
    ( \D_{\bar\y\bar\y}^\P )^{-1} 
    \d_{\bar\y\bar\x}^{\C}$.
    By substituting these results in  \eqref{eq:diff_mse}, we obtain that the MSE gap is 
\be
\label{MSE_gap}
    \varepsilon_{GSP-L}^2-\varepsilon_{GSP-WL}^2 =
   (\d_{\bar\y\bar\x}^{\C})^H  ( \D_{\bar\y\bar\y}^\P )^{-1}
    \d_{\bar\y\bar\x}^{\C},
    \ee
which defines the performance gain obtained by using the GSP-WLMMSE estimator over the GSP-LMMSE estimator.

By comparing \eqref{opt_WLMMSE_GSP_C0} with \eqref{opt_LMMSE_GSP1}
   it can be seen that  the structure of the GSP-WLMMSE estimator in \eqref{opt_WLMMSE_GSP_C0}
is, in fact, the sum of two {\em{linear}} estimators: the conventional GSP-LMMSE estimator based on $\y$, $\hat{\x}^{(\text{GSP-LMMSE})}$, and the GSP-LMMSE estimator based on the complex conjugate of $\y$, $\y^*$, as the measurement vector, which equals  $\V {\Dmat}_{\bar{\y}\bar{\x}}^\C(\Dmat_{\bar{\y}\bar{\y}}^\Gammamat)^{-1}
    \bar{\y}^*$.
This structure can be
explained by noting that the circularity assumption, $\C_{\y\y} = \ve{0}$, implies that the
vectors $\y$ and $\y^*$ are uncorrelated. Thus, the Hilbert subspaces
generated by these vectors are orthogonal, and taking into account $\y^*$
does not change the term coming from $\y$ only, 
which is equal to the GSP-LMMSE estimator, $\hat{\x}^{(\text{GSP-LMMSE})}$.
Thus, a nonzero matrix ${\Cmat}_{{\y}{\x}}$  necessarily implies an increment of the estimation
performance  when using widely-linear estimators instead of linear estimators. This result is an extension of the results in \cite{Picinbono_Chevalier1995} for the WLMMSE estimator.

\subsection{Maximal Improper Case}
\label{max_improper_sub}
Let $\F$ be a deterministic $N\times N$
complex-valued unitary matrix such that $\F^H\F=\F\F^H=\Imat$.
Then, if the measurement vector $\y$ is maximally improper \cite{lameiro2019maximally}, i.e.  $\y=\F \y^*$ 
with probability 1, we obtain  
that 
\begin{equation}
\label{F1}
\Gammamat_{\y\y}=\EX[\y \y^T \F^H]=\C_{\y\y} \F^H \Rightarrow \Gammamat_{\bar\y\bar\y}=\C_{\bar\y\bar\y} \bar{\F}^H,
\end{equation}
where $\bar{\F}^H\define\V^T\F^H \V$.
and 
\begin{equation}
\label{F2}\Gammamat_{\x\y}=\EX[\x \y^T \F^H]=\C_{\x\y} \F^H \Rightarrow \Gammamat_{\bar\x\bar\y}=\C_{\bar\x\bar\y} \bar{\F}^H.
\end{equation}
It can be verified that $\bar{\F}$ is also a unitary matrix, i.e. $\bar{\F}^H\bar{\F}=\bar{\F}\bar{\F}^H=\Imat$.
Thus, if $\bar{\F}$ is also a diagonal matrix, then, 
by substituting \eqref{F1} in \eqref{rho_def}, we obtain
\begin{equation}
\label{rho_def_F}
    \rho_n =\frac{| [\D_{\bar\y\bar\y}^\C]_{n,n}|^2}{| [\D_{\bar\y\bar\y}^\C]_{n,n} [ \bar{\F}^H]_{n,n}|^2}=1, ~n=1,\ldots,N.
    \end{equation}
Thus,  Theorem \ref{Theorem1},  which
was developed under the assumption that $ \rho_n \neq 1$,
does not hold.
In this case, it can be shown (similar to the derivations in Appendix \ref{appA}) that the GSP-WLMMSE estimator is reduced to 
 the GSP-LMMSE estimator from \eqref{opt_LMMSE_GSP1}.
 Thus, for a maximal improper $\y$ with a diagonal matrix $\bar\F$, the GSP-WLMMSE estimator coincides with the GSP-LMMSE estimator.
 This is similar to the property for the WLMMSE estimator and the LMMSE estimator  \cite{Adali_Schreier_Scharf2011}.
 
 However, if $\bar\F$ is not a diagonal matrix, then it can be verified that 
 \eqref{F1} implies that 
 \be
 \label{DDD}
[\D_{\bar\y\bar\y}^\Gammamat]_{n,n}=[\C_{\bar\y\bar\y} \bar{\F}^H]_{n,n}
=\sum_{m=1}^N [\C_{\bar\y\bar\y}]_{n,m} [\bar{\F}^H]_{m,n}.
\ee
That is, the diagonal of $\Gammamat_{\bar\y\bar\y}$ is a function of all the elements of $\C_{\bar\y\bar\y}$, and not only a function of its diagonal elements. As a result,
by substituting \eqref{F1} in \eqref{rho_def} we obtain
   $\rho_n \neq 1 $.
Therefore, in this case the GSP-WLMMSE estimator may be different from the GSP-LMMSE estimator.
 This is in contrast to the results for the WLMMSE and LMMSE for non-graph signals (see in \cite{Adali_Schreier_Scharf2011} after Eq. (46)). 
 
 This  can be explained by the fact that,  while
 $\bar\y$ and $\bar{\y}^*$  carry  the same second-order information about $\x$  where  $\y$ is maximally improper (i.e. $\y=\F \y^*$ and $\bar\y=\bar\F \bar\y^*$), 
 the GSP  estimators  only use the diagonal of the second-order information in the graph frequency domain, which may be different for 
 $\F \y^*$ and $\y^*$, as demonstrated by \eqref{DDD}.
 Thus,  the limitation to outputs of graph filters (or, equivalently, to diagonal covariance matrices in the graph frequency domain)  changes the information that we can obtain from  $\bar\y$ and $\bar{\y}^*$ for cases where $\F$ is a non-diagonal unitary matrix.

\subsection{Real-Values Case}
\label{real_sub}
 If the signal $\x$ is real (but $\y$ is complex), we have $\Dmat_{\bar{\y}\bar{\x}}^\C=\Dmat_{\bar{\y}\bar{\x}}^\Gammamat=(\Dmat_{\bar{\x}\bar{\y}}^\Gammamat)^*$. This leads to the simplified expression of the GSP-WLMMSE estimator from \eqref{opt_WLMMSE_GSP}:
 \beqna
\label{opt_LMMSE_GSP_x_real}
    \hat{\x}^{(\text{GSP-WLMMSE})}  \hspace{5.75cm}
    \nonumber\\
    =2
     \V   {\Re}\left\{ \left(\R\Dmat_{\bar{\x}\bar{\y}}^\Gammamat(\Dmat_{\bar{\y}\bar{\y}}^\Gammamat)^{-1}
    +(\I-\R){\Dmat}_{\bar{\y}\bar{\x}}^\Gammamat ({\Dmat}_{\bar{\y}\bar{\y}}^\C)^{-1}
  \right)\bar\y \right\}.
\eeqna
 Since $ \V$ is a real matrix, the GSP-WLMMSE estimator of a real graph signal $\x$ in \eqref{opt_LMMSE_GSP_x_real} is real. In contrast, the GSP-LMMSE estimator from \eqref{opt_LMMSE_GSP1} is generally
complex, which may lead to absurd results.

\subsection{Relation to Scalar Widely-Linear Estimators}
\label{scalar_sub}
In \cite{hellings2019combining} it is proposed to keep a linear filter and apply scalar
widely-linear operations to the components of the filter output. That is, to 
minimize the \ac{mse} among scalar widely-linear estimators of the form 
\begin{equation}
\label{general_form_scalar_widely_lin}
    \hat{\x} = {\text{diag}}(\a)\G \y + {\text{diag}}(\b)\G \y^*. 
\end{equation}
The total number of  unknown  complex coefficients  is 
$N^2+2N$ for scalar widely-linear estimators
 in \eqref{general_form_scalar_widely_lin}  (given by $\G$, $\a$, and $\b$),  $2N$  for  the proposed widely-linear GSP estimator of the form in \eqref{WLMMSE_general_form} (given by $f_1(\ve\lambda)$  and $\f_2(\ve\lambda)$), 
$2 N^2$ in the case of 
general widely-linear estimators in the form of \eqref{general_form_widely_lin} (given by $\Hmat_1$ and $\Hmat_2$), and $N$ in the case of a  GSP linear  estimator in the form of \eqref{GSPlinear} (given by $\f_1(\ve\lambda)$).
Where the number of unknowns increases, we have more degrees of freedom  for each estimator, which decreases the \ac{mse}. On the other hand, this results in a more  demanding computation,  and 
may be impractical for implementation with sample-mean values, where the actual statistic is intractable
(see more in Subsection \ref{discussion_subsec}).

It can be seen that the GSP widely-linear estimator can be written in the graph frequency domain as 
\begin{eqnarray}
\label{WLMMSE_general_form2}
    \hat{\bar\x}= 
    {\text{diag}}(\f_1(\ve\lambda)) \bar\y +  {\text{diag}}(\f_2(\ve\lambda)) \bar{\y}^*.
\end{eqnarray}
Therefore, the proposed GSP widely-linear estimator can be interpreted as a scalar widely-linear estimator \cite{hellings2019combining} in the {\em{graph frequency domain}}, where 
$\G=\I$.

\subsection{Relation with the Conventional WLMMSE Estimator}
\label{relation_wlmse_sub}
In the general case, the GSP-WLMMSE estimator from Theorem \ref{Theorem1}  has a higher MSE than the WLMMSE estimator from \eqref{WLMSE_est}, since the GSP widely-linear estimators belong to a subset of the widely-linear estimators.
The following theorem states sufficient and necessary conditions for the proposed GSP-WLMMSE estimator to achieve the WLMMSE estimator.
\begin{Theorem} \label{claim_coincides}
The GSP-WLMMSE estimator from Theorem \ref{Theorem1}  coincides with the WLMMSE estimator from \eqref{WLMSE_est} if
\beqna
\label{coincides_Appendix_to_prove_1}
(\Gammamat_{\bar\x\bar\y}-{\C}_{\bar\x\bar\y}(\Gammamat_{\bar\y\bar\y}^{-1})^*\C_{\bar\y\bar\y}^*)\P_{\bar\y\bar\y}^{-1} \hspace{2.75cm}
   \nonumber\\
       =(
    \D_{\bar\x\bar\y}^{\Gammamat}  - 
   \D_{\bar\x\bar\y}^{\C}
     (\D_{\bar\y\bar\y}^{\Gammamat})^{-1} (\D_{\bar\y\bar\y}^{\C} )^* )
     ( \D_{\bar\y\bar\y}^\P )^{-1}
    \eeqna
    and
    \beqna
    \label{coincides_Appendix_to_prove_2}
       (\C_{\bar\x\bar\y}-\Gammamat_{\bar\x\bar\y}\Gammamat_{\bar\y\bar\y}^{-1}{\C}_{\bar\y\bar\y})(\P_{\bar\y\bar\y}^{-1})^*   \hspace{3cm}
        \nonumber\\=
       (
    (\D_{\bar\y\bar\x}^{\C}
    -
    \D_{\bar\y\bar\y}^{\C}(\D_{\bar\y\bar\y}^{\Gammamat})^{-*}(\D_{\bar\y\bar\x}^{\Gammamat} )^*
    )
    ( \D_{\bar\y\bar\y}^\P )^{-1},
\eeqna
 where we assume that  $\Gammamat_{\bar\y\bar\y}$, $\P_{\bar\y\bar\y}$, $\D_{\bar\y\bar\y}^{\Gammamat}$, and $\D_{\bar\y\bar\y}^\P$ are non-singular matrices. 
\end{Theorem}
\begin{proof}
By comparing  \eqref{WLMSE_est} and  \eqref{opt_WLMMSE_GSP2}, 
it can be verified that the GSP-WLMMSE estimator coincides with the WLMMSE estimator for any observation vector $\y$ if the second-order statistics of $\x$ and $\y$ satisfy
\beqna
\label{comp1}
(\Gammamat_{\x\y}-{\C}_{\x\y}(\Gammamat_{\y\y}^{-1})^*\C_{\y\y}^*)\P_{\y\y}^{-1} \hspace{3cm}
   \nonumber\\
       =
    \V(
    \D_{\bar\x\bar\y}^{\Gammamat}  - 
   \D_{\bar\x\bar\y}^{\C}
     (\D_{\bar\y\bar\y}^{\Gammamat})^{-1} (\D_{\bar\y\bar\y}^{\C} )^*)( \D_{\bar\y\bar\y}^\P )^{-1}\V^T
    \eeqna
    and
    \beqna
    \label{comp2}
    (\C_{\x\y}-\Gammamat_{\x\y}\Gammamat_{\y\y}^{-1}{\C}_{\y\y})(\P_{\y\y}^{-1})^*\hspace{3cm}
        \nonumber\\=
    \V
    (
    (\D_{\bar\y\bar\x}^{\C}
    -
    \D_{\bar\y\bar\y}^{\C}(\D_{\bar\y\bar\y}^{\Gammamat})^{-*}(\D_{\bar\y\bar\x}^{\Gammamat} )^*
    )
    ( \D_{\bar\y\bar\y}^\P )^{-1}.
\eeqna
By right and left multiplication of \eqref{comp1} and \eqref{comp2} by $\Vmat^T$ and $\Vmat$, respectively, and using $\Vmat^T\Vmat=\Vmat\Vmat^T=\Imat$  and the GFT definition in \eqref{GFT_def}, the conditions in \eqref{comp1} and \eqref{comp2} can be rewritten as in \eqref{coincides_Appendix_to_prove_1} and \eqref{coincides_Appendix_to_prove_2}, respectively.
\end{proof}
A sufficient condition for \eqref{coincides_Appendix_to_prove_1} and \eqref{coincides_Appendix_to_prove_2} to hold is that the matrices 
 $\Gammamat_{\bar\y\bar\y}$, $\Gammamat_{\bar\x\bar\y}$,
 ${\C}_{\bar\x\bar\y}$, $\C_{\bar\y\bar\y}$, and $\P_{\bar\y\bar\y}$
   are diagonal matrices. 
In the following theorem we present a special case for which these matrices are diagonal, and thus,
the GSP-WLMMSE estimator coincides with the WLMMSE estimator.
\begin{Theorem}\label{claim_graphical_Model}
The GSP-WLMMSE estimator coincides with the WLMMSE estimator if 
the measurement vector is the output of  two linear graph filters,
$ h_1(\cdot)$ and $ h_2(\cdot)$:
   \be
    \label{g_filter}
    \y = \Vmat h_1(\ve\lambda)\Vmat^T \x
    +\Vmat h_2(\ve\lambda)\Vmat^T \x^*
    +\n,
    \ee
and  $\Cmat_{\x\x}$ and  $\Gammamat_{\x\x}$ are  diagonalizable by the eigenvector matrix of the Laplacian, $\Vmat$, i.e. $\Cmat_{\bar{\x}\bar{\x}}$ and $\Gammamat_{\bar{\x}\bar{\x}}$ are diagonal matrices.
The noise term $\n$  is assumed to be a zero-mean vector that is also  uncorrelated in the graph frequency domain, i.e. $\Cmat_{\bar{\n}\bar{\n}}$ is a diagonal matrix. In addition, $\x$ and $\n$ are assumed to be statistically independent.
\end{Theorem}
\begin{IEEEproof}
The proof is given in Appendix \ref{graphical_Model_Appendix}.
\end{IEEEproof}
An interesting case is the noiseless case where  $ h_1(\ve\lambda)=\zerovec$, and the measurement model from   \eqref{g_filter}
is reduced to
\be
    \label{g_filter2}
    \y =
    \Vmat h_2(\ve\lambda)\Vmat^T \x^*.
    \ee
    In addition, it is  assumed that
$h_2(\ve\lambda)$ is 
invertible
 such that $h_2^{-1}(\ve\lambda)
     h_2(\ve\lambda)=\Imat $. 
In Appendix \ref{new_appendix}, it is shown that, for this case, the GSP-WLMMSE and GSP-LMMSE estimators are 
\begin{equation}
    \label{new_eq}
      \hat{\x}^{(\text{GSP-WLMMSE})}   =   
    \V  
   h_2^{-1}(\ve\lambda)
    \V^T\y^*
\end{equation}
and
 \begin{equation}
\label{opt_LMMSE_GSP1th_4}
    \hat{\x}^{(\text{GSP-LMMSE})}  =   \Vmat     \Dmat_{\bar{\x}\bar{\x}}^\Cmat ( 
    \D_{\bar{\x}\bar{\x}}^\Gammamat )^{-1}h_2^{-1}(\ve\lambda)\bar\y.
\end{equation}
By substituting  \eqref{g_filter2} in \eqref{new_eq}, it can be seen that in this case
\begin{equation}
    \label{new_eq2}
      \hat{\x}^{(\text{GSP-WLMMSE})}   =   
    \V  
   h_2^{-1}(\ve\lambda)
    \V^T(\Vmat h_2(\ve\lambda)\Vmat^T \x^*)^*
    =\x.
\end{equation}
That is,
the GSP-WLMMSE estimator provides a
zero-error estimation in this case.
On the other hand, by observing \eqref{opt_LMMSE_GSP1th_4} it can be seen that we obtain a zero-error estimator only if $\Dmat_{\bar{\x}\bar{\x}}^\Cmat=  
    \D_{\bar{\x}\bar{\x}}^\Gammamat $, i.e. only for a real vector $\x$, where $\x=\x^*$.
Thus, for this case with complex-valued signal $\x$, the GSP-LMMSE estimator in \eqref{opt_LMMSE_GSP1th_4}  provides a nonzero estimation error, and thus is less attractive than the GSP-WLMMSE estimator. 

\section{Numerical Examples}
\label{simulation_sec}

In this section,  we evaluate the  performance of the sample-mean versions of the  LMMSE, WLMMSE, GSP-LMMSE, and GSP-WLMMSE estimators, where the covariance matrices are replaced by their sample-mean versions (see Subsection \ref{sample_version}). These estimators are denoted here as sLMMSE, sWLMMSE, sGSP-LMMSE, and sGSP-WLMMSE estimators. 
In addition, given $M$ Monte-Carlo trials, the empirical MSE of an estimator is computed as
$
    \text{MSE} = \frac{1}{M} \sum_{m=1}^{M}\|\hat\x_m-\x_m\|^2$, 
where $\x_m$ is the vector at the $m$th trial, and $\hat\x_m$ is its estimate.   
In all simulations, we use $M=10,000$.  

\subsection{Example 1: Linear Graph-Filter-Based  System}
\label{Ex1_subsec}
In this synthetic example,
we consider a random graph with $N=100$ nodes. 
The considered topology is shown in Fig. \ref{fig:graph}.
\begin{figure}[h]
 	\centering
 	\centerline{\includegraphics[width=0.35\columnwidth]{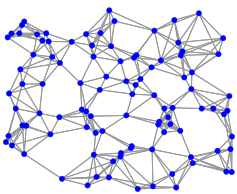}}
 	\caption{The considered network of Example 1 with $N=100$ sensors and an average of 350 edges.}
 	\label{fig:graph}
\end{figure}

The unknown graph signal $\x$ is given by \cite[p. 45]{AL10}
\begin{equation}
    \x = \sqrt{1-\eta^2}\ve\epsilon_r + j \eta \ve\epsilon_i,
\end{equation}
where $\ve\epsilon_r$ and $\ve\epsilon_i$ are two uncorrelated real-valued Gaussian random vectors with zero mean and unit variance. Thus, the covariance matrix and the complementary covariance matrix of $\x$ are $\Gammamat_{\x\x}=\I_N$ and $\C_{\x\x} = (1-2\eta^2)\I_N$, respectively.  The value of $\eta \in [0,1]$ controls the degree of non-circularity of the  vector $\x$, where for $\eta=1/\sqrt{2}$, $\x$ becomes circular.

The observed graph signal in this example is    
\begin{equation}\label{eq:sim lin filter}
    \y = \V \psi(\ve\Lambda) \V^T\x + \n,
\end{equation}
where $\n$ is a zero-mean proper Gaussian distributed  graph signal with a covariance matrix $\sigma^2 \I_N$ that is independent of $\x$, and $\psi(\ve\Lambda)$  is a graph filter, as defined in \eqref{graph_filter_def}.

It can be seen that the measurement model in
\eqref{claim_graphical_Model} is a special case of the model from \eqref{g_filter} with $ h_1=\psi$ and $ h_2=0$. In addition, since $\Cmat_{\x\x}$ and  $\Gammamat_{\x\x}$ are diagonal matrices, they are  diagonalizable by the Laplacian eigenvector matrix, $\Vmat$. Thus, the conditions of Theorem \ref{claim_graphical_Model} hold, and in this case, the  GSP-WLMMSE estimator coincides with the WLMMSE estimator. 

The theoretical covariance  and complementary covariance matrices of  $\bar\x$ and $\bar\y$ {
in the graph frequency domain are
 all diagonal matrices, i.e. $\Gammamat_{\bar\x\bar\y}=\D_{\bar\x\bar\y}^{\Gammamat} $, $\Gammamat_{\bar\y\bar\y}=\D_{\bar\y\bar\y}^{\Gammamat} $, 
$\C_{\bar\x\bar\y}=\D_{\bar\x\bar\y}^{\C} $, and $\C_{\bar\y\bar\y}=\D_{\bar\y\bar\y}^{\C}$.
where their diagonal entries  are
\begin{eqnarray}
\label{matrices_ex1}
  && [\Gammamat_{\bar\x\bar\y}]_{n,n} = \psi^*(\lambda_n),~ [\Gammamat_{\bar\y\bar\y}]_{n,n} = |\psi(\lambda_i)|^2 (1 + \rho_n), \nonumber \\&&
  ~ [\C_{\bar\x\bar\y}]_{n,n} = (1-2\eta^2)\psi(\lambda_n),\nonumber \\&&~ [\C_{\bar\y\bar\y}]_{n,n} =(1-2\eta^2)\psi^2(\lambda_n) , 
\end{eqnarray}
where we use the fact that $\V$ is a real unitary matrix. 
The parameter $\rho_n = \sigma^2/|\psi(\lambda_n)|^2$ can be interpreted as the \ac{snr} at the $n$th frequency of the graph signal.
 After a few algebraic steps, it can be shown 
that the GSP-WLMMSE graph filters according to \eqref{f1_paper} and \eqref{f2_paper} are 
\begin{eqnarray}
\label{f1_ex1}
    [\hat\f_1(\ve\lambda)]_n &=& \frac{1}{\psi(\lambda_n)} \frac{1}{1+\rho_n}  \frac{1-\frac{ (1-2\eta^2)^2}{1+\rho_n}}{1-(\frac{1-2\eta^2}{1+\rho_n})^2} \\
    \label{f2_ex1}
    \,[\hat\f_2(\ve\lambda)]_n &=& \frac{1}{\psi^*(\lambda_n)}\frac{1-2\eta^2}{1+\rho_n}  \frac{1-\frac{1}{1+\rho_n}}{1-(\frac{ 1-2\eta^2 }{1+\rho_n})^2},
\end{eqnarray}
while the filter of the GSP-LMMSE according to \eqref{opt_LMMSE_GSP1} is 
\begin{equation}
\label{f_ex1}
[\hat\f(\ve\lambda)]_n = \frac{1}{\psi(\lambda_n)} \frac{1}{1+\rho_n}.    
\end{equation}

It can be seen that for a proper graph signal, i.e. when $\eta=1/\sqrt{2}$, then $[\hat\f_1(\ve\lambda)]_n = [\hat\f(\ve\lambda)]_n$ and, as expected, $[\hat\f_2(\ve\lambda)]_n = 0$. Moreover, 
if the noise level is low, i.e. $\rho_n \cong 0$, then $[\hat\f_1(\ve\lambda)]_n = [\hat\f(\ve\lambda)]_n$ and $[\hat\f_2(\ve\lambda)]_n = 0$, for any $\eta\in[0,1]$.
Thus, there is no benefit in using the second filter of the GSP-WLMMSE in the noiseless case even if the graph signal is improper. The reason for this is that in the noiseless case, $[\hat\f_1(\ve\lambda)]_n$ perfectly equalizes the graph frequency response such that the reconstructed graph signal $\x$  is perfect, and thus the second filter is redundant. However, in general, by using \eqref{f1_ex1}-\eqref{f_ex1} and Theorem \ref{Theorem5}, the difference between the \ac{mse}s of the GSP-WLMMSE and the GSP-LMMSE estimator is
\begin{equation}
\label{var_eps_ex1}
    \varepsilon_{GSP-L}^2 - \varepsilon_{GSP-WL}^2
   =
     \sum_{n=1}^{N}\frac{(1-2\eta^2)^2}{1+\rho_n}  \frac{(1-\frac{1}{1+\rho_n})^2}{1-(\frac{ 1-2\eta^2 }{1+\rho_n})^2},
\end{equation}
where $\varepsilon_{GSP-L}^2 = \sum_{n=1}^{N} \rho_n/(1+\rho_n)$. For a proper graph signal, the \ac{mse}s of both filters are identical. However, in the extreme case when $\eta=0$, i.e. $\x =  \ve\epsilon_r$ or $\eta=1$, i.e. $\x = j\ve\epsilon_i$, then $\varepsilon_{GSP-WL}^2 =  \sum_{n=1}^{N} \rho_n/((1+\rho_n)(2+\rho_n))$.

In the simulations,  $\psi(\ve\Lambda)$ is chosen to be an 
ARMA graph filter, as in \eqref{ARMA2}, where the 
filter coefficients,
 $\{c_p\}_{p=0}^{P-1}$ and $\{a_q\}_{q=1}^Q$, are real-valued parameters, chosen randomly, and $P=Q=3$ are the filter orders. 
The values in \eqref{matrices_ex1} are used to calculate the theoretical MSEs of the theoretical  LMMSE, WLMMSE, GSP-LMMSE, and  GSP-WLMMSE estimators (i.e. with the true statistical covariance matrices) given by the mathematical terms 
$ \varepsilon_{L}^2$, $\varepsilon_{WL}^2$, 
$ \varepsilon_{GSP-L}^2$, and 
$\varepsilon_{GSP-WL}^2$
from \eqref{MSE_LMMSE}, \eqref{MSE_WLMSE}, \eqref{GSP_L_app}, and \eqref{GSP_WL_app2}, respectively.
These theoretical MSEs   and
the empirical MSEs of the sample-mean versions, the  sLMMSE, sWLMMSE, sGSP-LMMSE, and  sGSP-WLMMSE estimators,
are presented in Fig. \ref{fig:linear filter}  for $\eta=0.1,0.2,0.3,0.4$  versus the size of the training data of the sample-mean estimators,  $\Ntraining$ (see \eqref{one_sample}-\eqref{one_sample3}).
\begin{figure}[hbt]
 	\centering
 	\centerline{\includegraphics[width=0.7\columnwidth]{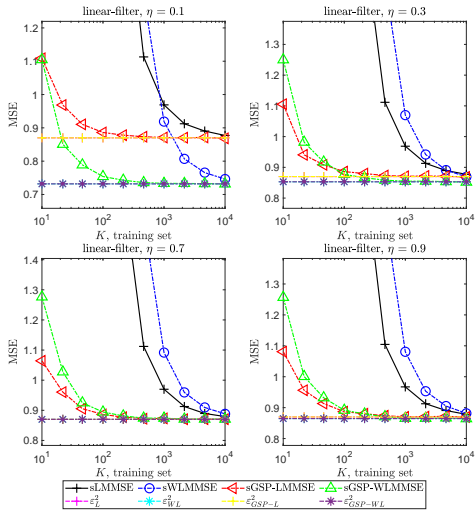}}
 	\caption{The theoretical MSEs, $ \varepsilon_{L}^2$, $\varepsilon_{WL}^2$, 
$ \varepsilon_{GSP-L}^2$, and 
$\varepsilon_{GSP-WL}^2$, and 
 	the 
 	empirical MSEs of  the sLMMSE, sWLMMSE, sGSP-LMMSE, and sGSP-WLMMSE estimators,
 	for the linear system of Example 1 versus the size of the training data,  $K$, for  $\eta=0.1,0.3,0.7$ and $0.9$.}
 	\label{fig:linear filter}
\end{figure}

A few insights   from this  experiment are as follows: \\
\renewcommand{\theenumi}{\arabic{enumi}}
    1) Since the MSEs of the LMMSE, WLMMSE, GSP-LMMSE, and  GSP-WLMMSE estimators, given by  $ \varepsilon_{L}^2$, $\varepsilon_{WL}^2$, $ \varepsilon_{GSP-L}^2$, and $\varepsilon_{GSP-WL}^2$, respectively, are based on the true, theoretical values of the covariance matrices and do not use the training data, their performance is independent of $\Ntraining$. It can be seen that  for a large enough $\Ntraining$, the empirical MSEs of the estimators coincide with the associated theoretical MSEs. Further, where $K$ is small ($K < 1,000$), the MSEs of both the sLMMSE and the sWLMMSE estimators diverge. This occurs as the number of unknowns of both estimators is approximately $N^2$, and therefore one should have $K > N^2 \approx 10,000$ to obtain a reliable non-parametric estimate of the entries of the sample covariance matrices defining the estimators. In fact, it is known in the literature that a reliable covariance matrix estimation is obtained if the number of samples is at least two to three times the matrix dimension  \cite{Reed1974Rapid}.\\
    2) On the other hand, as the sGSP-LMMSE and the sGSP-WLMMSE estimators are diagonals, we only need about $K \approx N$  to obtain reliable estimates. Thus,  the sGSP-WLMMSE and the sGSP-LMMSE estimators outperform the LMMSE and WLMMSE estimators  for low values of $K$, for any $\eta$, $N$.
    \\
   3) As the graph signal $\x$ becomes more improper, i.e. as $\eta$ decreases, the MSE of the sGSP-WLMMSE estimator becomes lower than the MSE of the sGSP-LMMSE estimator in accordance with the  theoretical study in Subsection \ref{discussion_subsec}.

We also evaluated the MSE of the different estimators as a function of the non-circularity coefficient, $\eta$, for $\eta\in[0,1]$, and for the sizes of the training data  $K=100$ and $K=10,000$. The results are shown in Fig. \ref{fig:mse vs eta}. Again, there are a few points we can conclude from this  experiment: \begin{figure}[h]
 	\centering
 	\centerline{\includegraphics[width=0.68\columnwidth]{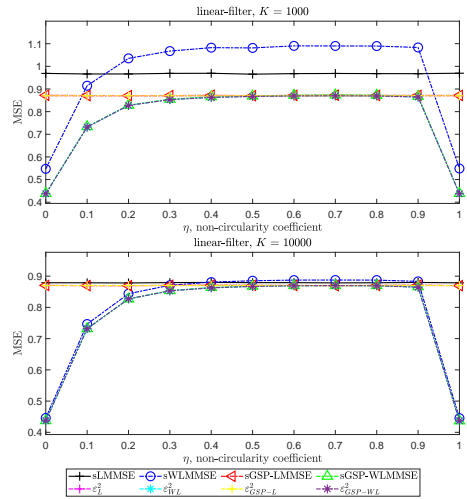}}
 	\caption{The theoretical MSEs, $ \varepsilon_{L}^2$, $\varepsilon_{WL}^2$, 
$ \varepsilon_{GSP-L}^2$, and 
$\varepsilon_{GSP-WL}^2$, and 
 	the 
 	empirical MSEs of  the sLMMSE, sWLMMSE, sGSP-LMMSE, and sGSP-WLMMSE estimators
 	for the linear system of Example 1 versus the non-circularity coefficient, $\eta$, and for $K=1,000,10,000$.}
 	\label{fig:mse vs eta}
\end{figure} 
\\
    1) Interestingly, for a small size of the training data, $K$, the empirical MSE of the sLMMSE estimator outperforms the MSE of the sWLMMSE estimator for almost any value of $\eta$. This occurs, as discussed above in Subsection \ref{sample_version}, due to the fact that the number of unknowns in the sWLMMSE estimator is larger than that of the sLMMSE estimator, and thus, the sLMMSE estimator provides a more reliable estimate compared to the WLMMSE estimator for a small  size of the training data.
    This is in contrast with the theoretical MSEs of the WLMMSE and LMMSE estimators  \cite{AL10}. Thus, this result emphasizes the need for  widely-linear estimators, but with fewer parameters to estimate, such as the GSP estimators.\\
    2) The empirical MSEs of the GSP estimators (GSP-LMMSE and  GSP-WLMMSE estimators) outperform the empirical MSEs of the conventional estimators (LMMSE and  WLMMSE estimators) for any value of $\eta$  for $K=1,000$ (small training set) and have similar empirical MSEs to their associated conventional estimators for $K=10,000$ (large training set). \\
   3) The MSE of the GSP-WLMMSE estimator is much lower than that of the GSP-LMMSE estimator for any value of the non-circularity coefficient, $\eta$, for both the theoretical and the empirical values. These MSEs almost coincide for  $\eta \approx 0.7$. This result can be explained by the fact that for this value the graph signal becomes almost proper, and there is no statistical benefit in using the complementary covariance matrix (see the theoretical discussion in Subsection \ref{proper_sub}).

\subsection{Example 2: Power System State Estimation (PSSE)}
\label{seeting_sec}
A power system can be represented as an undirected weighted graph, ${\pazocal{G}}({\pazocal{V}},\xi)$, where the set of vertices, $\pazocal{V}$, is the set of buses (generators or loads) and the edge set, $\xi$, is the set of transmission lines between these buses. 
The system (nodal) admittance matrix  is a $N\times N$ complex
symmetric matrix, where its
 $(m,k)$ element 
  is given by
 (p. 97 in \cite{Monticelli1999})
 \setlength{\belowdisplayskip}{5pt} \setlength{\belowdisplayshortskip}{5pt}
\setlength{\abovedisplayskip}{5pt} \setlength{\abovedisplayshortskip}{5pt}
\begin{equation}
\label{equ:YForm}
\displaystyle [\Ymat]_{m,k} = \begin{cases}
	\displaystyle -\sum_{n\in{\mathcal{N}}_m }  y_{m,n},& m = k \\
  \displaystyle  y_{m,k},& m\neq k,~(m, k) \in \xi 
\end{cases},
\end{equation}
and $0$ otherwise,  where  $y_{m,k}\in {\mathbb{C}}$ is the admittance of line $(m,k)$. 
The  admittance matrix from \eqref{equ:YForm}
satisfies
\be
\label{Y_def}
\Ymat=\G + j\B,
\ee
 where  the 
 conductance matrix, $\G$, and the  minus of the susceptance matrix, $-\B$, are  $N\times N$ real Laplacian matrices.
 
The commonly-used nonlinear AC power flow model \cite{Monticelli1999,Giannakis_Wollenberg_2013}, which is based on 
 Kirchhoff’s and Ohm’s laws, describes the relation between   the complex power injection and  voltage phasors. 
 According to this model,
the measurement vector of the active and reactive powers 
can be described by 
\be
\label{noisy_model}
\y={\text{diag}}(\x)\Ymat^{*}\x^{*}
+ \n,
\ee
where 
$\y \in {\mathbb{C}}^N$ is the power vector,
$\x\in {\mathbb{C}}^N$  is the voltage vector, and
the noise sequence, $\n$, is a   complex circularly symmetric Gaussian i.i.d. random vector with zero mean  and a known covariance matrix $ \sigma^2\I_N$, and is independent of $\x$.
In the graph modeling of the electrical network,
the Laplacian matrix, $\Lmat$, is usually constructed by using $-\B$ \cite{Grotas_Routtenberg_2019,dabush2021state}.
The goal of PSSE is to recover the state vector, $\x$, from the power measurements  $\y$ \cite{bienstock2019strong}, which is an NP-hard problem. 

The graph signal $\x$ is assumed to be given by
\begin{equation}\label{eq:x nonlin}
    \x = (1 + 0.1\ve\mu) \odot e^{j\ve\phi},
\end{equation}
where $\ve\mu \in \mathbb{R}^N$ is a zero-mean Gaussian  vector with covariance $\I_N$. In addition, $\ve\phi \in\mathbb{R}^N$ is a vector with i.i.d. entries, in which the $n$th entry is distributed according to $ \phi(n) \sim \mathcal{U}(0,\theta)$, where $\theta$ is the upper limit of the phase interval. Clearly, as $\theta$ (measured in radians) becomes closer to $2\pi$, the graph signal $\x$ becomes proper. This signal structure is used as it resembles the common voltage signal in AC power networks.

Since the optimal estimators and their  MSEs are intractable for this nonlinear case,  only the empirical MSEs of the sample-mean versions of the estimators (see Subsection \ref{sample_version}) are presented in the following results. 
The values of the physical parameters in \eqref{noisy_model} are taken from the test case of a 118-bus IEEE power system \cite{iEEEdata}, where $N=118$.
For a given $K$ and  a given phase limit, $\theta$,  
we generate the sLMMSE, sWLMMSE, sGSP-LMMSE, and  sGSP-WLMMSE estimators.  The empirical MSEs of the  sLMMSE, sWLMMSE, sGSP-LMMSE, and sGSP-WLMMSE estimators for the PSSE problem are presented  in Fig. \ref{fig:ieee118}  for different values of $\theta$ and  $K$.
\begin{figure}[h]
 	\centering
{\includegraphics[width=0.7\columnwidth]{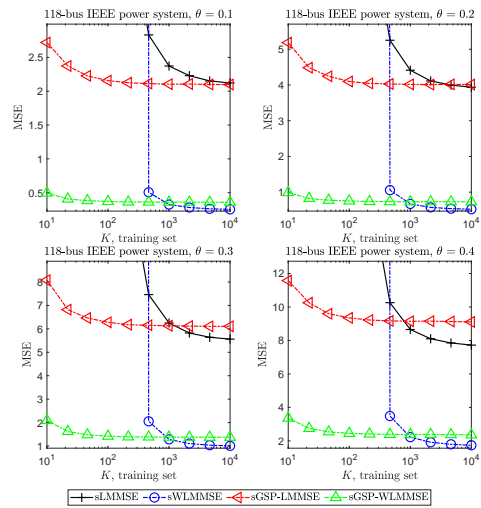}}
 	\caption{The empirical MSEs of the sLMMSE, sWLMMSE, sGSP-LMMSE, and sGSP-WLMMSE estimators versus $K$ for Example 2 with the 118-bus IEEE power system, for $\theta=0.1,0.2,0.3$, and $0.4$.}
 	\label{fig:ieee118}
\end{figure}

Here are some points we can conclude from this experiment: \\
    1) The MSEs of the sWLMMSE and the sGSP-WLMMSE estimators are lower than the MSEs of the sLMMSE and the sGSP-LMMSE estimators, especially for a large dataset (at least $K>30)$. Still, for a small $K$, the MSEs of the sLMMSE and the sWLMMSE estimators diverge and become worse than those of the widely-linear  (sGSP-LMMSE and  sGSP-WLMMSE) estimators. This size of the training set is in accordance with the conditions detailed in \cite{Reed1974Rapid} that ensure that the empirical covariance matrix is a reliable estimate of the theoretical covariance matrix.\\
    2) For any value of the phase $\theta$ and any $K$, the MSE of the sGSP-WLMMSE estimator  is lower than the MSEs of the sGSP-LMMSE and the sLMMSE estimators.\\
   3) Notwithstanding, the MSE of the sGSP-WLMMSE estimator is still larger than the MSE of the sWLMMSE estimator for any value of the maximal phase, $\theta$, and large $K$.    

We also examined the robustness of the estimators to model mismatch due to changes in the topology of the graph. We consider topology changes due to the removal or addition of edges in the graph. We assume that the noise standard deviation  is $\sigma=0.01$, and that the size of the training dataset is $K=1,000$. We also assume that the signal phase limit is $\theta=0.2$. The number of added or removed edges in the graph varies from 1 to 20. For each value of edges, we randomly selected the same number of nodes in the graph and removed (or added) a single edge from (to) their neighborhood. The MSE results are shown in Fig. \ref{fig:ieee118 robust}.  As can be seen, the sLMMSE estimator has the worst MSE performance in the case of edge removal (upper plot of Fig. \ref{fig:ieee118 robust}), while the sWLMMSE estimator has the worst MSE performance in the case of edge addition (lower plot of Fig. \ref{fig:ieee118 robust}). Both the sGSP-LMMSE and the sGSP-WLMMSE estimators are  robust and the change in their MSE is less affected by the changes in the graph topology compared with the sLMMSE and sWLMMSE estimators.
\begin{figure}[h]
 	\centering
{\includegraphics[width=0.65\columnwidth]{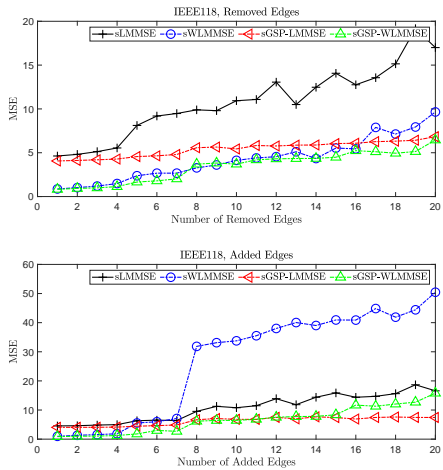}}
 	\caption{The empirical MSEs of the sLMMSE, sWLMMSE, sGSP-LMMSE, and sGSP-WLMMSE estimators versus the number of added/removed edges in the graph for $\theta=0.2$, and the size of  the training dataset  $K=1,000$.}
 	\label{fig:ieee118 robust}
\end{figure}

\section{Conclusion}
\label{conc}
In this paper, we  show that taking into consideration widely-linear systems in networks instead of strictly linear ones can yield significant improvements in the estimation of complex-valued graph signals. The conventional LMMSE and WLMMSE filters on graph signals  may involve a heavy computational burden and ignores the graph structure behind the data. Therefore, herein we suggest a reduced-complexity approach where the processing is based on graph filtering. In particular, we develop the GSP-WLMMSE estimator, which minimizes the MSE among the subset of widely-linear estimators that are represented as an output of two graph filters.
It is shown that the GSP-WLMMSE estimator of {\em{improper}} graph signals also requires the complementary auto-correlation and cross-correlation matrices. We investigate the properties of the GSP-WLMMSE estimator, its behavior under some special cases, and the conditions for the GSP-WLMMSE estimator to coincide with the WLMMSE estimator or with the GSP-LMMSE estimator. If the distributions of the graph signal and the observations are intractable, the sample-mean versions of the different estimators can be used. The diagonal structure of the sample-mean estimators, sGSP-LMMSE and sGSP-WLMMSE, in the graph frequency domain, bypasses the requirement for an extensive dataset in order to obtain stable sample-mean estimators, sLMSSE and sWLMMSE.

In simulations, we examine a synthetic example and the practical problem of state estimation in power systems. We show that for these problems, the proposed sample-mean version of the GSP-WLMMSE estimator achieves lower MSE than the linear estimators 
and the WLMMSE estimator for a limited training dataset. For sufficiently large datasets, the sGSP-WLMMSE estimator coincides with the sample-WLMMSE estimator, and they both outperform the linear estimators. Thus,   widely-linear GSP estimation can yield significant improvements in the MSE compared to strictly linear graph filters, when the circularity assumption does not hold.
Future work includes extensions to  nonlinear estimators,
implementation with graph filter parameterization to increase the robustness to   topology changes \cite{kroizer2021bayesian}, and investigations of the advantages of distributed implementations of the proposed estimator.
In addition, the development of estimators for time-varying complex-valued graph processes in a nonlinear dynamic system is of high interest. This can be done by integrating results from {\em{real-valued}} GSP-Kalman filters \cite{Isufi2020,Routtenberg_Sagi_ICASSP}
 and widely-linear complex Kalman filters \cite{dini2012class}.

\vspace{-0.15cm}
\appendices
	\renewcommand{\thesectiondis}[2]{\Alph{section}:}
 \section{Proof of Claim \ref{new_claim}}
 \label{appA_GSP_LMMSE}
In this appendix, we develop the GSP-LMMSE estimator, which is  the estimator that minimizes the \ac{mse}  over the subset of linear GSP estimators in the form of \eqref{GSPlinear}. For the sake of simplicity, in this appendix we replace $\f(\ve\lambda)$ by $\f$.
By substituting \eqref{GSPlinear} in the \ac{mse} definition, one obtains
\begin{eqnarray}
\label{min_f_linear}
\EX[\|\hat{\x} - \x\|^2]=
     \EX[\|\V {\text{diag}}(\f(\ve\lambda)) \V^T \y  - \x)\|^2] 
     \nonumber\\
     =
     \EX[\|{\text{diag}}(\f) \bar\y  - \bar\x\|^2],
     \hspace{1.5cm}
\end{eqnarray}
where $\bar\x$ and  $\bar\y$ are the GFTs of $\x$ and  $\y$, respectively, as defined in \eqref{GFT_def}. We also use the fact that $\V$ is a unitary matrix, i.e. $\V^T\V=\V \V^T=\Imat$. 
Using the fact that $\f$ is assumed to be a deterministic vector, as well as the properties of multiplications of diagonal matrices, we get
\begin{eqnarray}
\label{min_f_linear2}
\EX[\|\hat{\x} - \x\|^2]
     =\f^H  \Dmat_{\bar{\y}\bar{\y}}^\Gammamat\f-\ve{1}^T\Dmat_{\bar{\y}\bar{\x}}^\Gammamat - \f^H(\Dmat_{\bar{\x}\bar{\y}}^\Gammamat)^H\ve{1}+\bar\x^H\bar\x
     ,
\end{eqnarray}
where $\Dmat_{\bar{\x}\bar{\y}}^\Gammamat$ and $\Dmat_{\bar{\y}\bar{\y}}^\Gammamat $ are defined in \eqref{d_def1}.
Since \eqref{min_f_linear2} is a convex function of $\f$, 
the linear optimal graph filter is obtained by equating the derivative of the \ac{mse} in \eqref{min_f_linear2}  w.r.t. $\f^H$ to zero (by using complex derivative rules, see e.g., Appendix 2 in \cite{schreier2010statistical}), which results in
\begin{equation}
\label{new_normal}
\Dmat_{\bar{\y}\bar{\y}}^\Gammamat\hat{\f} = (\Dmat_{\bar{\x}\bar{\y}}^\Gammamat)^H\ve{1},
\end{equation}
where we substitute the definitions from \eqref{d_def1}.
Under the assumption that $\Dmat_{\bar{\y}\bar{\y}}^\Gammamat$ is a non-singular matrix, we obtain that
\begin{equation}
\label{new_normal3}
\hat{\f}  =(\Dmat_{\bar{\y}\bar{\y}}^\Gammamat)^{-1}
    (\Dmat_{\bar{\x}\bar{\y}}^\Gammamat)^H\ve{1}= 
    (\Dmat_{\bar{\y}\bar{\y}}^\Gammamat)^{-1}\Dmat_{\bar{\y}\bar{\x}}^\Gammamat\ve{1}
    =\Dmat_{\bar{\y}\bar{\x}}^\Gammamat(\Dmat_{\bar{\y}\bar{\y}}^\Gammamat)^{-1}\ve{1},
\end{equation}
where the last equality is since the multiplication of diagonal matrices can be reordered. By substituting \eqref{new_normal3} in \eqref{GSPlinear} we obtain the GSP-LMMSE estimator in \eqref{opt_LMMSE_GSP1}.

\section{Proof of Theorem \ref{Theorem1}} \label{appA}
In this appendix, we develop the GSP-WLMMSE estimator, which is  the estimator that minimizes the \ac{mse}  over the subset of widely-linear GSP estimators in the form of \eqref{WLMMSE_general_form}.
For the sake of simplicity, in this appendix we replace $\f_i(\ve\lambda)$ by $\f_i$, $i=1,2$.
By substituting \eqref{WLMMSE_general_form} in the \ac{mse} definition, one obtains
\begin{eqnarray}
\label{min_f1_f2}
\EX[\|\hat{\x} - \x\|^2]\hspace{5.5cm}
  \nonumber\\=
     \EX[\|\V({\text{diag}}(\f_1) \bar\y + {\text{diag}}(\f_2) \bar{\y}^*  - \V^T\x)\|^2]
     \nonumber\\
     =
     \EX[\|{\text{diag}}(\f_1) \bar\y + {\text{diag}}(\f_2) \bar{\y}^*  - \bar\x\|^2]
     ,\hspace{1.1cm}
\end{eqnarray}
where  $\bar\x$, $\bar\y$, and $\bar{\y}^*$ 
are the GFTs of $\x$, $\y$ and ${\y^*}$, respectively,
as defined in \eqref{GFT_def}. 
We also use the fact that $\V$ is a unitary matrix, i.e. $\V^T\V=\Imat$.
Using the properties of multiplications of diagonal matrices, the \ac{mse} in \eqref{min_f1_f2} can be written as 
    \begin{eqnarray}
    \label{epsilon1}
    \EX[\|\hat{\x} - \x\|^2]\hspace{6.25cm}\nonumber\\=
    \f_1^H {\text{ddiag}}(\EX[\bar\y^*\bar\y^T])\f_1 
    + \f_2^H {\text{ddiag}}(\EX[\bar\y\bar\y^H]) \f_2 
    \hspace{1.4cm}
    \nonumber\\
    + \f_2^H {\text{ddiag}}(\EX[\bar\y\bar\y^T]) \f_1
    + \f_1^H {\text{ddiag}}(\EX[\bar\y^*\bar\y^H])\f_2\hspace{1.3cm}
     \nonumber\\
    -\ve{1}^T {\text{ddiag}}(\EX[\bar\x^* \bar\y^T ]) \f_1 - \ve{1}^T {\text{ddiag}}(\EX[\bar\x^*\bar\y^H]) \f_2 \hspace{1.2cm}\nonumber\\
    - \f_1^H {\text{ddiag}}(\EX[\bar\y^*\bar\x^T])\ve{1} - \f_2^H {\text{ddiag}}(\EX[\bar\y\bar\x^T]) \ve{1}
     +\EX[ \bar\x^H\bar\x]     
    .
\end{eqnarray}
Since \eqref{epsilon1} is a convex function of the pair $\f_1,\f_2\in  \mathbb{C}^{N}$, 
the optimal graph filters are obtained by equating the derivative of the \ac{mse} in \eqref{epsilon1} w.r.t. $\f_1,\f_2$ to zero (using complex derivative rules, see e.g., Appendix 2 in \cite{schreier2010statistical}), which results in
\begin{equation}
\label{normal_equation}
    \left[
    \begin{array}{cc}
      (\D_{\bar\y\bar\y}^{\Gammamat})^*    &   (\D_{\bar\y\bar\y}^{\C} )^*\\
      \D_{\bar\y\bar\y}^{\C}   & \D_{\bar\y\bar\y}^{\Gammamat}
    \end{array}
    \right]
    \left[
    \begin{array}{c}
        \hat{\f}_1  \\
        \hat{\f}_2
    \end{array}
    \right] 
    =
    \left[
    \begin{array}{c}
        (\d_{\bar\y\bar\x}^{\Gammamat} )^*  \\
        \d_{\bar\y\bar\x}^{\C}
    \end{array}
    \right],
\end{equation}
where $\d_{\bar\y\bar\x}^{\Gammamat}\define
{\text{diag}}(\Gammamat_{\bar{\y}\bar{\x}})
$ and $\d_{\bar\y\bar\x}^{\C} \define
{\text{diag}}(\C_{\bar{\y}\bar{\x}})
$, and $  \D_{\bar\y\bar\y}^{\Gammamat}$ and  $\D_{\bar\y\bar\y}^{\C} $
are defined in \eqref{d_def1} and \eqref{d_def2}, respectively. 
Applying the matrix inversion lemma on \eqref{normal_equation} leads to 
        

\begin{eqnarray}
    \hat{\f}_1 &=& 
    ( \D_{\bar\y\bar\y}^\P )^{-1
    }
    (\d_{\bar\y\bar\x}^{\Gammamat} )^* -( \D_{\bar\y\bar\y}^\P )^{-1} (\D_{\bar\y\bar\y}^{\C} )^* (\D_{\bar\y\bar\y}^{\Gammamat})^{-1} \d_{\bar\y\bar\x}^{\C} \nonumber\\
    &=&
    ( \D_{\bar\y\bar\y}^\P )^{-1}(
    (\d_{\bar\y\bar\x}^{\Gammamat} )^* - 
    (\D_{\bar\y\bar\y}^{\C} )^* (\D_{\bar\y\bar\y}^{\Gammamat})^{-1} \d_{\bar\y\bar\x}^{\C})
   \eeqna
   and
   \beqna
   \label{f2app}
    \hat{\f}_2 &=& 
    -( \D_{\bar\y\bar\y}^\P)^{-1}
    \D_{\bar\y\bar\y}^{\C}(\D_{\bar\y\bar\y}^{\Gammamat})^{-*}(\d_{\bar\y\bar\x}^{\Gammamat} )^* 
    +
    ( \D_{\bar\y\bar\y}^\P)^{-1}
    \d_{\bar\y\bar\x}^{\C} \nonumber\\
    &=&
    ( \D_{\bar\y\bar\y}^\P )^{-1}
    (
    \d_{\bar\y\bar\x}^{\C}
    -
    \D_{\bar\y\bar\y}^{\C}(\D_{\bar\y\bar\y}^{\Gammamat})^{-1}(\d_{\bar\y\bar\x}^{\Gammamat} )^*
    ),
\end{eqnarray}
where $\D_{\bar\y\bar\y}^{\Gammamat}$ is a real matrix, and the Schur complement  of the left matrix in \eqref{normal_equation}, $\D_{\bar\y\bar\y}^\P$, is defined 
in \eqref{Schur_diag} and is a real diagonal matrix 
(thus, we used $(\D_{\bar\y\bar\y}^\P)^*=\D_{\bar\y\bar\y}^\P$ and   $(\D_{\bar\y\bar\y}^\P)^{-H}=  (\D_{\bar\y\bar\y}^\P )^{-1}$).
Equivalently, 
if $| [\Gammamat_{\bar\y\bar\y}]_{n,n} |^2 \neq | [\C_{\bar\y\bar\y}]_{n,n}|^2$, then 
the $n$th elements of  $\hat{\f}_2$ and  $\hat{\f}_1$ satisfy
\begin{eqnarray}
\label{hat_f1_n}
    \hat{f}_{1,n}
    =
    \frac{ ([\Gammamat_{\bar\y\bar\x}]_{n,n})^*  ([\Gammamat_{\bar\y\bar\y}]_{n,n}) - [\C_{\bar\y\bar\x}]_{n,n}   ([\C_{\bar\y\bar\y}]_{n,n})^*}
    {| [\Gammamat_{\bar\y\bar\y}]_{n,n} |^2 - | [\C_{\bar\y\bar\y}]_{n,n}|^2}, 
  \end{eqnarray}
\begin{eqnarray}
\label{hat_f2_n}
    \hat{f}_{2,n}
    =
    \frac{[\C_{\bar\y\bar\x}]_{n,n} [ \Gammamat_{\bar\y\bar\y}]_{n,n}  -  ([\Gammamat_{\bar\y\bar\x}]_{n,n})^*  [{\C}_{\bar\y\bar\y}]_{n,n}}
    {|[\Gammamat_{\bar\y\bar\y}]_{n,n}|^2 - | [\C_{\bar\y\bar\y}]_{n,n}|^2},
\end{eqnarray}
for any $n=1,\ldots,N$.
Now, if $ \rho_n \neq 1$, where  $\rho_n$ is defined in \eqref{rho_def}, then
\eqref{hat_f1_n} and \eqref{hat_f2_n} can be written as 
\begin{eqnarray}
\label{hat_f1_n_with_rho}
    \hat{f}_{1,n}
    =
    \frac{1}{1 - \rho_n}
    \frac{([\Gammamat_{\bar\y\bar\x}]_{n,n})^*}{ [\Gammamat_{\bar\y\bar\y}]_{n,n}} 
    +\left(1-\frac{1}{1 - \rho_n}\right)
    \frac{[\C_{\bar\y\bar\x}]_{n,n}}
    {[{\C}_{\bar\y\bar\y}]_{n,n}}
    \end{eqnarray}
    \begin{eqnarray}
    \hat{f}_{2,n}
   =
    \frac{1}{1-\rho_n}\frac{[\C_{\bar\y\bar\x}]_{n,n}}{([\Gammamat_{\bar\y\bar\y}]_{n,n})^*}
    +
    \left(1-\frac{1}{1 - \rho_n}\right)
    \frac{([\Gammamat_{\bar\y\bar\x}]_{n,n})^*}{([\C_{\bar\y\bar\y}]_{n,n})^*}.
    \label{hat_f2_n_with_rho}
\end{eqnarray}
By substituting \eqref{hat_f1_n_with_rho} and \eqref{hat_f2_n_with_rho} in \eqref{WLMMSE_general_form},
we obtain 
the GSP-WLMMSE estimator in \eqref{opt_WLMMSE_GSP}.

\section{Orthogonality  principle}
\label{app_calc}
In this appendix we develop the orthogonality  principle of the GSP-WLMMSE estimator. 
Using
\eqref{opt_WLMMSE_GSP2}, we obtain
\beqna
\label{ort1}
\EX[(\hat\x^{(\text{GSP-WLMMSE})}-\x)^H\y]
=-\EX[\x^H\y]
 \hspace{1cm}\nonumber
   \\
    +\EX[\bar{\y}^H 
    ( \D_{\bar\y\bar\y}^\P )^{-1}
    (
    \D_{\bar\x\bar\y}^{\Gammamat}  - 
    (\D_{\bar\y\bar\y}^{\C} )^* (\D_{\bar\y\bar\y}^{\Gammamat})^{-1} \D_{\bar\y\bar\x}^{\C})^H \bar\y]\nonumber\\
    +\EX[\bar\y^T  
    ( \D_{\bar\y\bar\y}^\P )^{-1}
    (
    (\D_{\bar\y\bar\x}^{\C}
    -
    \D_{\bar\y\bar\y}^{\C}(\D_{\bar\y\bar\y}^{\Gammamat})^{-1}(\D_{\bar\y\bar\x}^{\Gammamat} )^*
    )^H
    \bar\y],
\eeqna  
where we used \eqref{GFT_def} and the fact that $\V$ and $\D_{\bar\y\bar\y}^{\Gammamat}$ are real matrices.
By using the trace operator properties and  $\EX[\x^H\y]=\EX[\x^H\Vmat\Vmat^T\y]=\tr(\Gammamat_{ \bar\y \bar\x})$,  \eqref{ort1} can be rewritten as 
\beqna
\label{ort2}
\EX[(\hat\x^{(\text{GSP-WLMMSE})}-\x)^H\y]
=-\tr(\Gammamat_{ \bar\y \bar\x})
 \hspace{1cm}\nonumber
   \\ \tr \left(
    ( \D_{\bar\y\bar\y}^\P )^{-1}
    (
    \D_{\bar\x\bar\y}^{\Gammamat}  - 
    (\D_{\bar\y\bar\y}^{\C} )^* (\D_{\bar\y\bar\y}^{\Gammamat})^{-1} \D_{\bar\y\bar\x}^{\C})^H \Gammamat_{ \bar\y \bar\y}\right)\nonumber\\
     + \tr \left( 
     ( \D_{\bar\y\bar\y}^\P )^{-1}
    (
    (\D_{\bar\y\bar\x}^{\C}
    -
    \D_{\bar\y\bar\y}^{\C}(\D_{\bar\y\bar\y}^{\Gammamat})^{-*}(\D_{\bar\y\bar\x}^{\Gammamat} )^*
    )^H \C_{\bar\y\bar\y} \right)
    \nonumber\\
    =-\tr(\Dmat_{\bar{\y}\bar{\x}}^\Gammamat )
 \hspace{1cm}\nonumber
   \\ \tr \left(
    ( \D_{\bar\y\bar\y}^\P )^{-1}
    (
    \D_{\bar\x\bar\y}^{\Gammamat}  - 
    (\D_{\bar\y\bar\y}^{\C} )^* (\D_{\bar\y\bar\y}^{\Gammamat})^{-1} \D_{\bar\y\bar\x}^{\C})^H \Dmat_{\bar{\y}\bar{\y}}^\Gammamat \right)\nonumber\\
     + \tr \left( 
     ( \D_{\bar\y\bar\y}^\P )^{-1}
    (
    (\D_{\bar\y\bar\x}^{\C}
    -
    \D_{\bar\y\bar\y}^{\C}(\D_{\bar\y\bar\y}^{\Gammamat})^{-*}(\D_{\bar\y\bar\x}^{\Gammamat} )^*
    )^H \Dmat_{\bar{\y}\bar{\y}}^\C \right),
\eeqna
where the last equality is obtained by using the definitions in \eqref{d_def1} and \eqref{d_def2}, and since $\tr(\D \A)=\tr(\D {\text{ddiag}}(\Amat))$ for any diagonal matrix $\D$.
By substituting \eqref{Schur_diag} in \eqref{ort2}, it can be verified that
$\EX[(\hat\x^{(\text{GSP-WLMMSE})}-\x)^H\y]=0$.
Similarly, it can be shown that $\EX[(\hat\x^{(\text{GSP-WLMMSE})}-\x)^H\y^*]=0$.


\section{Proof of Theorem \ref{Theorem5}} \label{appB}
In this appendix, we develop the explicit terms of the \ac{mse} of the GSP-LMMSE estimator from \eqref{opt_LMMSE_GSP1} and  of the GSP-WLMMSE estimator from \eqref{opt_WLMMSE_GSP}.
 The proof is along the lines of the proof of 
 \cite[Eq. (1.39)]{AL10} for the WLMSE estimator. In this appendix  we replace $\f_i(\ve\lambda)$ by $\f_i$ for $i=1,2$.

According to  the orthogonality principle of the GSP-WLMMSE estimator, developed in Appendix \ref{app_calc}, $\EX[(\hat\x^{(\text{GSP-WLMMSE})}-\x)^H\y]=0$ and $\EX[(\hat\x^{(\text{GSP-WLMMSE})}-\x)^H\y^*]=0$. In other words, the error vector of the GSP-WLMMSE estimator is orthogonal to any linear combination of these vectors. Since the estimate $\hat\x^{(\text{GSP-WLMMSE})}$ is a linear combination of $\y$ and $\y^*$ then $(\hat\x^{(\text{GSP-WLMMSE})}-\x)\perp\hat\x^{(\text{GSP-WLMMSE})}$, and as a result
\begin{equation}
    \EX[\|\hat\x^{(\text{GSP-WLMMSE})})\|^2]=\EX[(\hat\x^{(\text{GSP-WLMMSE})})^H\x].
\end{equation}
Thus,  the \ac{mse} of the GSP-WLMMSE estimator is given by
\begin{eqnarray}
\label{MSE_app2}
    \varepsilon_{GSP-WL}^2 =\EX[\|\hat{\x}^{(\text{GSP-WLMMSE})} - \x\|^2]
    \nonumber\hspace{1.1cm}\\= \EX[\|\x\|^2] - \EX[\|\hat\x^{(\text{GSP-WLMMSE})}\|^2].
\end{eqnarray}
By substituting the estimator $\hat\x^{(\text{GSP-WLMMSE})}$, which  is a widely-linear GSP
estimator described by
\eqref{WLMMSE_general_form} with 
the optimal graph frequency responses $\hat{\f}_1$ and $\hat{\f}_2$ from Appendix \ref{appA},
and after a few algebraic steps, the \ac{mse} in \eqref{MSE_app2} is also expressed in a compact form as
\begin{eqnarray}
\label{GSP_WL_app}
    \varepsilon_{GSP-WL}^2 =
    \tr(\Gammamat_{\x\x}) \hspace{4.8cm}\nonumber\\
    -
    \left[
    \begin{array}{c}
         \hat{\f}_1  \\
         \hat{\f}_2 
    \end{array}
    \right]^H
     \left[
    \begin{array}{cc}
         (\D^\Gammamat_{\bar\y\bar\y})^* & (\D^\C_{\bar\y\bar\y})^* \\
         \D^\C_{\bar\y\bar\y} & \D^\Gammamat_{\bar\y\bar\y}
    \end{array}
    \right]
     \left[
    \begin{array}{c}
         \hat{\f}_1  \\
         \hat{\f}_2 
    \end{array}
    \right]
    \nonumber\\
        =
    \tr(\Gammamat_{\x\x})
    -
    \left[
    \begin{array}{c}
        (\d_{\bar\y\bar\x}^{\Gammamat} )^*  \\
        \d_{\bar\y\bar\x}^{\C}
    \end{array}
    \right]^H\hspace{2.3cm}\nonumber\\
    \qquad \times
     \left[
    \begin{array}{cc}
         \D^\Gammamat_{\bar\y\bar\y} & (\D^\C_{\bar\y\bar\y})^* \\
         \D^\C_{\bar\y\bar\y} & \D^\Gammamat_{\bar\y\bar\y}
    \end{array}
    \right]^{-1}
    \left[
    \begin{array}{c}
        (\d_{\bar\y\bar\x}^{\Gammamat} )^*  \\
        \d_{\bar\y\bar\x}^{\C}
    \end{array}
    \right],  
\end{eqnarray}
where the last equality is obtained by substituting \eqref{normal_equation} and using $(\D^\Gammamat_{\bar\y\bar\y})^*=\D^\Gammamat_{\bar\y\bar\y}$.
By using the inverse of a block matrix, \eqref{GSP_WL_app} can be rewritten as 
\begin{eqnarray}
\label{GSP_WL_app2}
    \varepsilon_{GSP-WL}^2 =
    \tr(\Gammamat_{\x\x})
    -
    \left[
    \begin{array}{c}
        (\d_{\bar\y\bar\x}^{\Gammamat} )^*  \\
        \d_{\bar\y\bar\x}^{\C}
    \end{array}
    \right]^H\hspace{2.3cm}\nonumber\\
    \qquad \times
     \left[
    \begin{array}{lr}
      ( \D_{\bar\y\bar\y}^\P )^{-1} & \Emat_{1,2} \\
         \Emat_{2,1} &   ( \D_{\bar\y\bar\y}^\P )^{-1}
    \end{array}
    \right]
    \left[
    \begin{array}{c}
        (\d_{\bar\y\bar\x}^{\Gammamat} )^*  \\
        \d_{\bar\y\bar\x}^{\C}
    \end{array}
    \right],  
\end{eqnarray}
where the off-diagonal blocks are 
\begin{eqnarray}
\label{E12def}
    \Emat_{12} &=&  
    -( \D_{\bar\y\bar\y}^\P )^{-1} 
    (\D_{\bar\y\bar\y}^{\C} )^* (\D_{\bar\y\bar\y}^{\Gammamat})^{-1}\\
    \Emat_{21} &=& 
    -( \D_{\bar\y\bar\y}^\P)^{-1}
    \D_{\bar\y\bar\y}^{\C}(\D_{\bar\y\bar\y}^{\Gammamat})^{-1}
  ,\label{E21def}
\end{eqnarray}
and  $\D_{\bar\y\bar\y}^\P$    is a real diagonal matrix defined in \eqref{Schur_diag}.

On the other hand, by using \eqref{opt_LMMSE_GSP1}, it can be shown that the \ac{mse} of the GSP-LMMSE estimator is
\begin{eqnarray}
    \varepsilon_{GSP-L}^2
    =\EX[\|\hat{\x}^{(\text{GSP-LMMSE})} - \x\|^2]\hspace{1cm}
    \nonumber\\=
    \tr(\Gammamat_{\x\x})
    - (\d_{\bar\y\bar\x}^{\Gammamat})^T
    (\D^\Gammamat_{\bar\y\bar\y})^{-*}
    (\d_{\bar\y\bar\x}^{\Gammamat})^*,
\end{eqnarray}
which can also be expressed in an equivalent way as
\begin{eqnarray}
\label{GSP_L_app}
    \varepsilon_{GSP-L}^2 =
    \tr(\Gammamat_{\x\x})\hspace{4.5cm}  \nonumber\\
    -
    \left[
    \begin{array}{c}
        (\d_{\bar\y\bar\x}^{\Gammamat} )^*  \\
        \d_{\bar\y\bar\x}^{\C}
    \end{array}
    \right]^H
    \left[
    \begin{array}{cc}
         (\D^\Gammamat_{\bar\y\bar\y})^{-*} & \ve{0} \\
         \ve{0} & \ve{0}
    \end{array}
    \right]
    \left[
    \begin{array}{c}
        (\d_{\bar\y\bar\x}^{\Gammamat} )^*  \\
        \d_{\bar\y\bar\x}^{\C}
    \end{array}
    \right].
\end{eqnarray}
The difference between
the  MSEs from \eqref{GSP_WL_app} and \eqref{GSP_L_app}
is 
\beqna
\label{gap_MSE}
    \varepsilon_{GSP-L}^2 - \varepsilon_{GSP-WL}^2 =
    \left[
    \begin{array}{c}
        (\d_{\bar\y\bar\x}^{\Gammamat} )^*  \\
        \d_{\bar\y\bar\x}^{\C}
    \end{array}
    \right]^H \hspace{2.5cm}
    \nonumber\\  
    \times
    \left[
    \begin{array}{cc}
          ( \D_{\bar\y\bar\y}^\P )^{-1}- (\D^\Gammamat_{\bar\y\bar\y})^{-1}& \Emat_{1,2} \\
        \Emat_{2,1} &  ( \D_{\bar\y\bar\y}^\P )^{-1}
    \end{array}
    \right]
    \left[
    \begin{array}{c}
        (\d_{\bar\y\bar\x}^{\Gammamat} )^*  \\
        \d_{\bar\y\bar\x}^{\C}
    \end{array}
    \right].
\eeqna
By using \eqref{ppp}, it can be verified that
\beqna
\label{E11}
( \D_{\bar\y\bar\y}^\P )^{-1}- (\D^\Gammamat_{\bar\y\bar\y})^{-1}\hspace{3.5cm}
\nonumber\\
=
    (\D_{\bar\y\bar\y}^{\Gammamat})^{-1}
(\D_{\bar\y\bar\y}^\C)^{*}
    (\D_{\bar\y\bar\y}^\P)^{-1}
\D_{\bar\y\bar\y}^{\C} (\D_{\bar\y\bar\y}^{\Gammamat})^{-1}\nonumber\\
=\Emat_{2,1} \D_{\bar\y\bar\y}^\P \Emat_{1,2},
\end{eqnarray}
where we used the fact that $\Emat_{2,1}, \D_{\bar\y\bar\y}^\P, \Emat_{1,2}$ are  based on   multiplications of diagonal matrices, and thus we can change the order  of the multiplications.
By substituting \eqref{E11} in \eqref{gap_MSE} and using the fact that $\Emat_{2,1}=\Emat_{1,2}^*$, we obtain that \begin{eqnarray}  
\label{gap}
    \varepsilon_{GSP-L}^2 - \varepsilon_{GSP-WL}^2
     \hspace{3.5cm} \nonumber\\=
     (\d_{\bar\y\bar\x}^{\Gammamat} )^T 
     \Emat_{2,1} \D_{\bar\y\bar\y}^\P \Emat_{1,2}
     (\d_{\bar\y\bar\x}^{\Gammamat} )^*
     +(\d_{\bar\y\bar\x}^{\Gammamat} )^T  \Emat_{1,2}
      \d_{\bar\y\bar\x}^{\C}
      \nonumber\\
    + ( \d_{\bar\y\bar\x}^{\C})^H \Emat_{2,1} (\d_{\bar\y\bar\x}^{\Gammamat} )^*
    +
        ( \d_{\bar\y\bar\x}^{\C})^H
        (\D_{\bar\y\bar\y}^\P)^{-1}\d_{\bar\y\bar\x}^{\C}
 \nonumber\\=
    (\d_{\bar\y\bar\x}^{\C} - (\D_{\bar\y\bar\y}^{\C} )^* (\D_{\bar\y\bar\y}^{\Gammamat})^{-1} (\d_{\bar\y\bar\x}^{\Gammamat} )^*)^H 
    ( \D_{\bar\y\bar\y}^\P )^{-1} 
    \nonumber\\
    \times(\d_{\bar\y\bar\x}^{\C} - (\D_{\bar\y\bar\y}^{\C} )^* (\D_{\bar\y\bar\y}^{\Gammamat})^{-1} (\d_{\bar\y\bar\x}^{\Gammamat} )^*), 
\end{eqnarray}
where the last equality is obtained by substituting \eqref{E12def} and \eqref{E21def}.
Finally, by substituting \eqref{f2_paper} in \eqref{gap},
$\varepsilon_{GSP-L}^2 - \varepsilon_{GSP-WL}^2$ can be  expressed in a compact form as \eqref{eq:diff_mse}.


\section{Proof of Theorem \ref{claim_graphical_Model}} \label{graphical_Model_Appendix}
 By using \eqref{g_filter}, we obtain that
\begin{eqnarray} \label{covariance_matrix_app_2}
   \Gammamat_{\bar{\y}\bar{\y}}
    =
    \EX[\bar{\y}\bar{\y}^H]
    = h_1(\ve\lambda)
    \Gammamat_{\bar{\x}\bar{\x}} h_1^*(\ve\lambda)
    +
    h_2(\ve\lambda)
    \Gammamat_{\bar{\x}\bar{\x}}^* h_2^*(\ve\lambda)
    \nonumber\\
    + h_1(\ve\lambda) \Cmat_{\bar{\x}\bar{\x}}h_2^*(\ve\lambda)
    + h_2(\ve\lambda)\Cmat_{\bar{\x}\bar{\x}}^* h_1^*(\ve\lambda)
    +\Gammamat_{\bar{\n}\bar{\n}},
\end{eqnarray}
where we use the GFT definition in \eqref{GFT_def}.
Since according to the Theorem conditions, 
 $\Gammamat_{\bar{\n}\bar{\n}}$, $\Gammamat_{\bar{\x}\bar{\x}}$, and $\Cmat_{\bar{\x}\bar{\x}}$ are diagonal matrices, we obtain that  $ \Gammamat_{\bar{\y}\bar{\y}}$ is a diagonal matrix, i.e.
\begin{equation}\label{graphical_Model_Appendix_C_xx_f}
    \Gammamat_{\bar{\y}\bar{\y}} = 	\Dmat_{\bar{\y}\bar{\y}}^\Gammamat.
\end{equation}
Similarly, the following covariance matrices are all diagonal: 
\begin{eqnarray} \label{covariance_matrix_app_2_C}
   \Cmat_{\bar{\y}\bar{\y}}
    =
    \EX[\bar{\y}\bar{\y}^T]
    = h_1(\ve\lambda)
    \Cmat_{\bar{\x}\bar{\x}} h_1(\ve\lambda)
    +
    h_2(\ve\lambda)
    \Cmat_{\bar{\x}\bar{\x}}^* h_2(\ve\lambda)
    \nonumber\\
    + h_1(\ve\lambda) \Gammamat_{\bar{\x}\bar{\x}}h_2(\ve\lambda)
    + h_2(\ve\lambda)\Gammamat_{\bar{\x}\bar{\x}}^* h_1(\ve\lambda)
    +\Cmat_{\bar{\n}\bar{\n}},
\end{eqnarray}
\begin{eqnarray} \label{graphical_Model_Appendix_C_tx_A}
   \Gammamat_{\bar\x\bar\y} = 
    \EX[\bar\x\bar\y^H]
=  \Gammamat_{\bar{\x}\bar{\x}} h_1^*(\ve\lambda)
+\Cmat_{\bar{\x}\bar{\x}}h_2^*(\ve\lambda),
\end{eqnarray}
\begin{eqnarray} \label{graphical_Model_Appendix_C_tx_B}
   \Cmat_{\bar\x\bar\y} = 
    \EX[\bar\x\bar\y^T]
=  \Cmat_{\bar{\x}\bar{\x}} h_1(\ve\lambda)
+\Gammamat_{\bar{\x}\bar{\x}}h_2(\ve\lambda).
\end{eqnarray}
As a result, the Schur complement from \eqref{Pyy} satisfies
\beqna
\label{Pyy_}
    \P_{\y\y} = {\text{ddiag}}(\Gammamat_{\y\y}) \hspace{4.5cm}\nonumber\\-  {\text{ddiag}}(\C_{\y\y}) {\text{ddiag}}((\Gammamat_{\y\y}^{-1})^*) {\text{ddiag}}(({\C}_{\y\y})^*)
  =
    \D_{\bar\y\bar\y}^\P,
\eeqna
where $ \D_{\bar\y\bar\y}^\P$ is defined in \eqref{Schur_diag}.
 Therefore, since all the matrices involved are diagonal matrices,  \eqref{coincides_Appendix_to_prove_1} and \eqref{coincides_Appendix_to_prove_2} hold.

\section{Derivation of   \eqref{new_eq} and \eqref{opt_LMMSE_GSP1th_4}}
\label{new_appendix}
In this appendix we derive
     \eqref{new_eq} and \eqref{opt_LMMSE_GSP1th_4} for the model in \eqref{g_filter2}, under the assumptions that 
$h_2(\ve\lambda)$ is 
invertible
 such that $h_2^{-1}(\ve\lambda)
     h_2(\ve\lambda)=\Imat $. 
     By substituting $\Gammamat_{\bar{\n}\bar{\n}} =\Cmat_{\bar{\n}\bar{\n}}=\zerovec$ and $ h_1(\ve\lambda)=\zerovec$ in  \eqref{covariance_matrix_app_2}-\eqref{graphical_Model_Appendix_C_tx_B} from  Appendix \ref{graphical_Model_Appendix}, and using the fact that $(\D_{\bar{\x}\bar{\x}}^\Gammamat)^*=\D_{\bar{\x}\bar{\x}}^\Gammamat$, we obtain  that in this case
\begin{eqnarray} \label{covariance_matrix_app_2_0}
   \Gammamat_{\bar{\y}\bar{\y}}
    =
    h_2(\ve\lambda)
    \Gammamat_{\bar{\x}\bar{\x}}^* h_2^*(\ve\lambda)\Rightarrow \D_{\bar\y\bar\y}^{\Gammamat}= 
    h_2(\ve\lambda)
    \D_{\bar{\x}\bar{\x}}^\Gammamat h_2^*(\ve\lambda),
\end{eqnarray}
\begin{eqnarray} \label{covariance_matrix_app_2_C_0}
   \Cmat_{\bar{\y}\bar{\y}}
    =
    h_2(\ve\lambda)
    \Cmat_{\bar{\x}\bar{\x}}^* h_2(\ve\lambda)\Rightarrow \D_{\bar{\y}\bar{\y}}^\Cmat
    =
    h_2(\ve\lambda)
    (\D_{\bar{\x}\bar{\x}}^\Cmat)^*h_2(\ve\lambda),
\end{eqnarray}
\begin{eqnarray} \label{graphical_Model_Appendix_C_tx_0}
   \Gammamat_{\bar\x\bar\y} = 
    \Cmat_{\bar{\x}\bar{\x}}h_2^*(\ve\lambda)\Rightarrow \D_{\bar\x\bar\y}^\Gammamat = 
    \Dmat_{\bar{\x}\bar{\x}}^\Cmat h_2^*(\ve\lambda),
\end{eqnarray}
and
\begin{eqnarray} \label{graphical_Model_Appendix_C_tx_00}
   \Cmat_{\bar\x\bar\y} 
=  \Gammamat_{\bar{\x}\bar{\x}}h_2(\ve\lambda)\Rightarrow \Dmat_{\bar\x\bar\y}^\Cmat 
=  \D_{\bar{\x}\bar{\x}}^\Gammamat h_2(\ve\lambda).
\end{eqnarray}
By substituting \eqref{covariance_matrix_app_2_0}-\eqref{graphical_Model_Appendix_C_tx_00}
in \eqref{opt_WLMMSE_GSP2}-\eqref{Schur_diag}, we obtain 
\begin{equation}
\label{Schur_diag_sc4}
    \D_{\bar\y\bar\y}^\P=
        h_2(\ve\lambda) \left( 
    \D_{\bar{\x}\bar{\x}}^\Gammamat 
     -     
    (\D_{\bar{\x}\bar{\x}}^\Cmat)^*(\D_{\bar{\x}\bar{\x}}^\Gammamat)^{-1}
    \D_{\bar{\x}\bar{\x}}^\Cmat 
    \right)  h_2^*(\ve\lambda)
\end{equation}
and
\begin{eqnarray}
\label{opt_WLMMSE_GSP2_th4}
    \hat{\x}^{(\text{GSP-WLMMSE})}  \hspace{5.75cm}\nonumber
   \\
    =
    \V(
     \Dmat_{\bar{\x}\bar{\x}}^\Cmat  - 
  \D_{\bar{\x}\bar{\x}}^\Gammamat 
     (  \D_{\bar{\x}\bar{\x}}^\Gammamat)^{-1}   \D_{\bar{\x}\bar{\x}}^\Cmat )h_2^*(\ve\lambda)
     ( \D_{\bar\y\bar\y}^\P )^{-1}\bar\y\hspace{0.75cm} \nonumber\\
    +
    \V  h_2(\ve\lambda)
    (
    \D_{\bar{\x}\bar{\x}}^\Gammamat
    -
   (\D_{\bar{\x}\bar{\x}}^\Cmat)^* (\D_{\bar{\x}\bar{\x}}^\Gammamat)^{-1}\D_{\bar\x\bar\x}^{\C} 
    )
    ( \D_{\bar\y\bar\y}^\P )^{-1}
    \bar\y^*\nonumber\\
    =   
    \V  
    \D_{\bar\y\bar\y}^\P h_2^{-1}(\ve\lambda) ( \D_{\bar\y\bar\y}^\P )^{-1}
    \bar\y^*
    =   
    \V  
   h_2^{-1}(\ve\lambda)
    \bar\y^*,\hspace{1.5cm}
\end{eqnarray}
where we used the fact that for any diagonal matrices with appropriate dimensions $\D_A,\D_B,\D_C$, we have $\D_A\D_B\D_C=\D_B\D_C\D_A$.
The last equality gives the result in \eqref{new_eq}.
 On the other hand, 
by substituting
 \eqref{covariance_matrix_app_2_0} and \eqref{graphical_Model_Appendix_C_tx_0}
in \eqref{opt_LMMSE_GSP1}, we obtain that the GSP-LMMSE
 estimator for this case is given by \eqref{opt_LMMSE_GSP1th_4}.
 \bibliographystyle{IEEEtran}


\end{document}